\documentclass[11pt]{article}
\usepackage{a4wide}
\usepackage{amsmath}
\usepackage{amssymb}
\usepackage{amsthm}
\usepackage{dsfont}
\usepackage{graphicx}
\usepackage{soul}
\usepackage{xcolor}
\usepackage{comment}
\usepackage{algorithm}
\usepackage{algpseudocode}
\usepackage{bbold}
\usepackage{appendix}
\usepackage{todonotes}

\usepackage{hyperref}

\newtheorem{theorem}{Theorem}
\newtheorem{lemma}[theorem]{Lemma}

\newtheorem{claim}[theorem]{Claim}
\newtheorem{definition}[theorem]{Definition}

\renewcommand{\epsilon}{\varepsilon}

\title{Better Trees for Santa Claus}

\author{\'Etienne Bamas \vspace{0.3mm} \\  EPFL \and Lars Rohwedder \vspace{0.3mm} \\  Maastricht University}

\date{}

\begin{document}
\maketitle
\thispagestyle{empty}

\begin{abstract}
    We revisit the problem max-min degree
    arborescence, which was introduced by
    Bateni et al.\ [STOC'09] as a central special case
    of the general Santa Claus problem, which constitutes
    a notorious open question in approximation algorithms.
    In the former problem we are given a directed graph with sources and sinks and our goal is to find
    vertex disjoint arborescences rooted in the sources such that at each non-sink vertex
    of an arborescence the out-degree is at least $k$, where $k$ is to be maximized.
    
    This problem is of particular interest, since it appears to capture much of the difficulty of the
    Santa Claus problem:
    (1) like in the Santa Claus problem the configuration LP has a large integrality
    gap in this case and
    (2) previous progress by Bateni et al.\ was quickly generalized to the Santa Claus problem
    (Chakrabarty et al.\ [FOCS'09]).
    These results remain the state-of-the-art both for the Santa Claus problem and for max-min degree arborescence and they yield a polylogarithmic approximation in quasi-polynomial time.
    We present an exponential improvement to this,
    a $\mathrm{poly}(\log\log n)$-approximation in quasi-polynomial time for the max-min degree arborescence problem. To the best of our knowledge, this is the first example of breaking the logarithmic barrier for a special case of the Santa Claus problem, where
    the configuration LP cannot be utilized.
    
    The main technical novelty of our result
    are locally good solutions: informally, we
    show that it suffices to find a
    $\mathrm{poly}(\log n)$-approximation that
    locally has stronger guarantees.
    We use a lift-and-project type of LP and randomized rounding, which were also used by
    Bateni et al., but unlike previous work
    we integrate careful pruning steps in the rounding. In the proof we extensively apply Lov\'asz Local Lemma
    and a local search technique, both of which were previously used only in the context of the configuration LP.
\end{abstract}

\newpage
\section{Introduction}
In the Santa Claus problem (also known as max-min fair allocation),
there are gifts that need to be assigned
to children. Each gift $j$ has unrelated values $v_{ij}$ for each of
the children $i$. The goal is to assign each gift $j$ to a child $\sigma(j)$
such that we maximize the utility of the least happy child, that is,
$\min_i \sum_{j : \sigma(j) = i} v_{ij}$.
The dual of the problem, where one has to minimize the maximum
instead of maximizing the minimum is the problem of makespan
minimization on unrelated parallel machines. Both variants 
form notoriously difficult open problems in approximation algorithms~\cite{bansal2017scheduling, schuurman1999polynomial, williamson2011design}
and there is a common believe that the Santa Claus problem
admits a constant approximation if and only if makespan minimization on unrelated machines admits a better-than-$2$
approximation~\cite{bansal2017scheduling}. Although formally no such reduction is known,
techniques often seem to transfer from one problem to the
other, see for example the restricted assignment case
in the related work section.

Bateni, Charikar, and Guruswami~\cite{bateni2009maxmin}
identified as a central special case of the Santa Claus problem
the restriction that for all values we have $v_{ij}\in \{0,1,\infty\}$
and for each gift there is at most one child with $v_{ij} = \infty$ and for each child there is at most one gift with $v_{ij} = \infty$.
This case can be rephrased as the following graph problem,
dubbed the max-min degree arborescence problem.
We are given a directed graph $G = (V, E)$, a set of sources
$S\subseteq V$, and a set of sinks $T\subseteq V$.
Our goal is to compute a set of vertex disjoint arborescences
rooted in each of the sources $S$. The leaves of these
arborescences are in the sinks in~$T$.
For all inner vertices
we have an out-degree of $k$, where $k$ is to be maximized.
Bateni et al.\ gave a $\max\{\mathrm{poly}(\log n), n^{\epsilon}\}$-approximation in
time $n^{O(1/\epsilon)}$, which is still the state-of-the-art for this problem. In particular, they obtain
a quasi-polynomial time polylogarithmic approximation by
setting $\epsilon = \log\log n/\log n$. In a highly non-trivial way,
the same approach was then generalized by Chakrabarty, Chuzhoy, and Khanna~\cite{chakrabarty2009allocating} to obtain
the same result also for the Santa Claus problem.
The only hardness known for the Santa Claus problem is that there is no better-than-2 approximation (see \cite{lenstra1990approximation,bezakova2005allocating}) and it is not difficult to show there is no better-than-$\sqrt{e/(e-1)}$ approximation for the max-min degree arborescence problem via a reduction from max-$k$-cover (see Appendix \ref{sec:hardness}). This
still leaves a large gap in the understanding of both problems.
In particular, the two most pressing questions are
whether the polylogarithmic guarantee can also be achieved
in polynomial time and whether a sublogarithmic approximation
guarantee can be achieved. In this paper we answer the latter
question affirmatively for the max-min degree arborescence
problem by giving a $\mathrm{poly}(\log\log n)$-approximation
in quasi-polynomial time.

While much progress on sublogarithmic approximations
has been made on other special cases
of the Santa Claus problem, most notable the restricted assignment case, these results are all based on the configuration LP. On the other hand, it is known that
for the general Santa Claus problem the configuration LP
has an unbounded integrality gap and hence these methods
seem unlikely to generalize. Similarly, the max-min
degree arborescence problem is a case where the integrality
gap of the configuration LP is already high~\cite{bateni2009maxmin} and as such the
algorithmic techniques need to be rethought. Given this and
the fact that previous progress on the problem
was quickly extended to the general Santa Claus problem
we believe that the techniques introduced in this
paper are highly relevant towards the goal of
understanding the approximability of the Santa Claus problem.

\subsection{Other related work}
An influential line of work on the Santa Claus problem addresses the so-called restricted assignment case, which is the case where $v_{ij}\in \{0,v_j\}$ for any child $i$ and gift $j$. Here, $v_j$ is a uniform value of the gift. One
may also rephrase this as the setting where all gifts
have the same value for each child, but they cannot be
assigned to all of them.
In a seminal work, Bansal and Srividenko~\cite{bansal2006santa} provide a $O(\log \log (n)/\log \log \log (n))$-approximation algorithm for this case using randomized rounding of the configuration LP using Lov\'asz Local Lemma.
This was improved by Feige~\cite{feige2008allocations} to an $O(1)$-approximation.
Further progress on the constant or the running time was made since then,
see e.g.~\cite{annamalai2017combinatorial, davies2020tale, Polacek, DBLP:journals/talg/AsadpourFS12, DBLP:conf/icalp/ChengM18, DBLP:conf/icalp/ChengM19, haxell2022improved}, or extended to more general valuation functions \cite{bamas2021submodular}. Many of these results are based on a local search technique due to Haxell~\cite{haxell1995condition}, who developed it in the context of hypergraph matching.
The technique was later also adapted to the dual problem (makespan minimization) by Svensson \cite{svensson2012santa} to give a better-than-2 approximation in the restricted assignment case of it as well. This result also led to a series of improvements either in the running time or the approximation factor \cite{jansen2020quasi, chakrabarty20141, annamalai2019lazy, jansen2018estimating}. Unfortunately, all these results rely on the configuration LP; hence they seem unlikely to extend to the max-min degree arborescence problem or even the general Santa Claus problem.
Perhaps surprisingly, the local search technique still plays a secondary, but important role in our proof, where we use
it without the configuration LP.
Also Lov\'asz Local Lemma is important in our proof,
although our use differs significantly from that
of Bansal and Srividenko.

The Graph Balancing version of the Santa Claus problem is the special case where every gift has a non-zero value only for two children. Although this variant also has a graph structure,
it is very different in nature from max-min degree arborescence.
This problem is resolved with a remarkably clean polynomial time $2$-approximation~\cite{verschae2014configuration} and
a matching lower bound. It is quite surprising that this
works even when the values for the two children are unrelated,
because the corresponding min-max version with unrelated values
has not seen any progress so far (for the min-max version where the value is the same for both children see~\cite{EbenlendrKS14, JansenR19}). Nevertheless, also the Graph Balancing special
case of the Santa Claus problem has the striking feature that the configuration LP is strong, see for example~\cite{chakrabarty2009allocating}.

Another problem that is at least intuitively related to the max-min arborescence is the directed Steiner Tree problem.
In this problem we are given an edge-weighted directed graph with a source and a set of sinks. The goal is to find an arborescence of minimal weight, which is
rooted at the source and spans all sinks.
It is quite remarkable that the state-of-the-art for this
problem is very similar to ours: there is a $n^{\epsilon}$-approximation algorithm in polynomial time for every fixed $\epsilon > 0$
and a polylogarithmic approximation algorithm in quasi-polynomial time~\cite{charikar1999approximation}. Unlike our problem, it was
shown that no sublogarithmic (in fact, no $\log^{2-\epsilon} n$) approximation exists~\cite{halperin2003polylogarithmic}.
It may therefore come as a surprise
that we can indeed find a sublogarithmic approximation for our
problem.
Related to the directed Steiner Tree problem is also
the (undirected) group Steiner Tree problem, which can
be shown to be a special case.
In this problem, we are given an undirected weighted graph and a list of groups that are subsets of vertices. The goal is to compute the cheapest set of edges that is connected and contains at least one vertex from each group. 
Here, there are more subtile connections to our problem:
taking a closer look at the literature one can notice that the challenging instances in the group Steiner Tree problem have a similar structure to the challenging ones for our case. More precisely they are layered graphs with $O(\log n)$ layers. Halperin et al.~\cite{halperin2007integrality} show that the integrality gap of a natural LP relaxation of group Steiner Tree could be amplified from $\Theta (\log(n))$ on $O(1)$-layered instances to $\tilde \Theta (\log^2(n))$ on $\Omega (\log(n))$-layered graphs. This construction was later transformed to the hardness result by Halperin and Krauthgamer~\cite{halperin2003polylogarithmic}. 

Before our work, it was quite plausible that such an amplification technique could also apply in the context of the max-min degree arborescence problem. This would have shown how to amplify a $\Omega(1)$ gap on $O(1)$-layered instances to a $\tilde \Omega(\log(n))$ gap on $\Theta(\log (n))$-layered instances. In Appendix \ref{sec:hardinstance}, we adapt the construction of \cite{halperin2007integrality} to our setting. At first sight, the gap indeed seems to amplify and our construction shows that the previous rounding algorithms of \cite{chakrabarty2009allocating, bateni2009maxmin} cannot hope to get better than a $\tilde \Omega (\log(n))$-approximation. Fortunately, we notice that a single round of top-to-bottom pruning (throwing away half the children of every vertex) seems to resolve that issue. We note that the group Steiner Tree problem has a rich history and it would be very interesting to see if more techniques could be transferred to the max-min degree arborescence problem.

\subsection{Informal overview of techniques}
A crucial idea that goes back to previous work~\cite{bateni2009maxmin, chakrabarty2009allocating}
is to allow congestion in the solution.
Generally, a vertex can only have one incoming edge
in the solution, but we relax this constraint.
We call the maximum number of times a vertex is used
the congestion.
The algorithms in~\cite{bateni2009maxmin, chakrabarty2009allocating} employ randomized rounding
to obtain a solution with polylogarithmic congestion.
The congestion can then be translated into an approximation
rate by relatively straight-forward arguments.
This polylogarithmic congestion comes from a Chernoff
bound that yields an inversely polynomial probability, which is then applied to all vertices with a union bound.

A new ingredient of our algorithm is the notion of local congestion.
Roughly speaking, we first compute a solution, which
still has polylogarithmic congestion, but when considering
only a local part of the solution (say, vertices within
a distance of $\ell = O(\log\log n)$ in the arborescences),
then this local part needs to have much smaller congestion,
i.e., $\mathrm{poly}(\log\log n)$. In other words,
if a vertex is used multiple times in the solution, then
the occurrences should be far apart in the arborescences.

First, let us describe why it is plausible to be
able to obtain such a guarantee. It is already known
that by randomized rounding a polylogarithmic congestion
can be achieved. The local congestion on the other hand
is by definition a very local constraint and hence
Lov\'asz Local Lemma (LLL) is natural to employ.
This is indeed our approach, although the details
are challenging.

Next, we will explain how to arrive at a sublogarithmic
congestion. The approach that we call top-to-bottom
pruning is very blunt: slighly oversimplifying, we
take a given solution (with $\mathrm{poly}(\log n)$ congestion) and start at the sources of the arborescences. We throw away randomly a constant fraction (say, half) of their children.
Then we move to the other children and recurse.
Clearly, this decreases the approximation rate only
insignificantly. However, in expectation the congestion
at each vertex decreases drastically. If for example
a vertex is at distance $d$ from the source of an arborescence, then the probability of it not being removed
is only $1/2^d$. There is, however, a caveat here:
suppose that
the same vertex occurs in a (relaxed) arborescence many times
and all occurrences are very close to each other. Then
there is a high positive correlation between the vertices
not being removed, which forms a serious problem.
Indeed, this is where the local
congestion comes in. It essentially bounds the dependence
of occurrences surviving. Again, this proof makes use of
LLL, because it seems infeasible to try and make the probability of a vertex's congestion staying above $\mathrm{poly}(\log\log n)$ small enough to apply a union bound.

Since both parts require LLL, it is crucial to bound
the dependencies. However, if $k$ is large (say $\Omega(n)$), then even a local part of the arborescence
contains many vertices. It then seems unlikely to be
able to guarantee locally low congestion for every
vertex. Roughly speaking, we will only guarantee
the property for a large fraction of the vertices, so
that we can make the probability inversely polynomial
in $k$. All other vertices need to be removed from the solution and this is generally very dangerous: even if we remove
only a small fraction of vertices, this can lead to other vertex removals becoming necessary, because they now
have a low out-degree. If we are not careful, this can
accelerate and corrupt the whole solution.
We call this the bottom-to-top pruning and we formalize
a condition, under which the damage to the solution can
be controlled. This condition is then applied in both parts.

This brings us to a discussion on our two pruning techniques. Intuitively, both approaches have complementary merits to each other. Bottom-to-top pruning makes it easy to maintain a low maximum congestion, but difficult to keep a good number of children for every vertex in the arborescence. On the contrary, top-to-bottom pruning makes it very easy to maintain a good number of children, but difficult to keep the maximum congestion under control. Our proof can be seen as a careful combination of those two techniques using LLL. The use of pruning makes the proof fairly involved and one might wonder if this could not be avoided. In particular, it is not clear if the analysis of the previous randomized rounding algorithm (see \cite{bateni2009maxmin,chakrabarty2009allocating}) is tight or not. However, we argue that an $\Omega (\log n/\log \log n)$ factor seems unavoidable in previous works, and that it is not easily fixable. We now elaborate on this: In previous approaches as in ours, a crucial part of the algorithm is to solve the max-min degree arborescence problem on layered instances (that is, the sources are located in the first layer, and edges can exist only between vertices of consecutive layers). Previous works~\cite{bateni2009maxmin,chakrabarty2009allocating} then solve these instances by rounding an LP relaxation of the problem (we use the same LP relaxation but with a different rounding). An intuitive randomized rounding that appears in previous works is roughly as follows. Assume that the LP says there exists a solution of value $k$. Then the source samples $k$ children with probability equal to the LP values. These selected children then select $k$ children each, again equal to the LP values
(more precisely, values that correspond to conditioning on the previous selections). In that manner, we make progress layer by layer until reaching the last layer. This guarantees a maximum congestion that is at most polylogarithmic in the number of vertices, hence the polylogarithmic approximation ratio. At this point, one might be tempted to argue that very few vertices in our solution will have congestion $\Omega (\log n)$ and that they are not a serious problem. Unfortunately, we show in Appendix \ref{sec:hardinstance} an instance in which the above rounding results in a solution in which \emph{all} the sinks selected in the arborescence have an expected congestion of $\Omega (\log n/\log \log n)$. While this may seem counter-intuitive, recall that here we are implicitly conditioning by the fact of being selected in the solution. It is non-trivial to recover from this issue: For instance deleting---in a bottom-to-top fashion---the vertices with high congestion will basically remove almost all the sinks which will corrupt the whole solution. In fact, we argue in Appendix \ref{sec:bottom-to-top_pruning} that the bottom-to-top pruning is very sensitive to a small number of deletions if one aims at a sublogarithmic approximation ratio. We believe this is another reason why getting such a guarantee is challenging.

Lastly, there is one important issue that we have not mentioned so far: in a similar way that
each vertex loses some fraction of its children in the rounding (compared to the LP relaxation),
we would lose some fraction of the sources. This happens also in previous works~\cite{bateni2009maxmin,chakrabarty2009allocating}, who then only compute a solution for a $1/\mathrm{poly}(\log n)$ fraction of the sources and repeat it for $\mathrm{poly}(\log n)$ times to cover all sources. This again introduces a congestion of $\mathrm{poly}(\log n)$, which seems difficult to avoid with the randomized rounding approach. First, we only design the randomized rounding to approximate single source instances so that
we do not have to cope with this. Then we present a black-box reduction from many sources to one source, which
uses a non-trivial machinery that is very different from the randomized rounding approach.
Namely, this is by a local search framework, which
has already seen a big application in
the restricted assignment case of the Santa Claus problem, see related work.
Usually the local search is analyzed against the configuration LP, which is not applicable here, 
so our way of applying it is quite different to previous works.

\subsection{Relevance to the Santa Claus problem}
While it may seem as if the techniques introduced in this paper are limited to the 
max-min degree arborescence problem, we believe that they are of interest also towards obtaining
a better approximation for the Santa Claus problem. We emphasize that in the
state-of-the-art~\cite{chakrabarty2009allocating}, the authors arrive through a series
of sophisticated reductions at a similar problem as the layered instances of max-min bounded
arborescence (see next section for definition of layered instances). The main difference
is that they cannot simply select $k$ outgoing edges for each vertex independently.
Instead, the $k$ children have to be connected to the parent through a capacitated flow network shared
by all vertices. For details we refer to~\cite{chakrabarty2009allocating}.
An obstacle in adapting our arguments to that setting is that these flow networks introduce new
dependencies, which seem difficult to bound and our LLL arguments are naturally
very sensitive to dependencies.
On the other hand, the general idea of using locally low congestion solutions as an intermediate
goal and applying methods similar to the top-to-bottom and bottom-to-top pruning appear directly relevant and may not
be limited to the way we apply them with LLL.

\section{Overview}\label{sec:prelim}
We start with some simplifying assumptions.
Let $k$ be the optimum of the given instance,
which we guess via a binary search framework.
If $k\le \mathrm{poly}(\log\log n)$, then obtaining a $(1/k)$-approximation is sufficient. This is easy to achieve: it is enough to find $|S|$ vertex-disjoint paths that connect each source to a sink, which can be done by a standard max-flow algorithm. Throughout the paper we assume without loss of generality
that $k$ is at least $\mathrm{poly}(\log\log n)$ with sufficiently large constants. Similarly, we assume that
$n$ is larger than a sufficiently large constant.
Whenever we divide $k$ by some term, for example, $k / \log\log n$, we would normally have to write
$\lfloor k / \log\log n \rfloor$. All the divisors considered can be assumed to be much smaller than $k$, hence any loss due to rounding is only a small constant factor and therefore insignificant. We assume for simplicity that the divisors are always integral and
the divisions have no remainder. 

We describe solutions using a set of paths instead
of directly as arborescences.
This abstraction will become useful in the linear
programming relaxation, but also in other parts
throughout the paper.
With a given arborescence, we associate the
set of all paths from a source to
some vertex in it.
Let $p$ be a path from a source to some vertex $v$. Formally, $p$ is a tuple of vertices $(v_0,v_1,\ldots, v_\ell)$ where $v_0$ is a source, $v_1$ a vertex with an edge from $v_0$, and more generally $v_i$ is a vertex
having an edge from $v_{\ell-1}$. 
We do not explicitly forbid circles, but
the properties of a solution will imply that only simple
paths can be used. Finally we say that $p$ is a \emph{closed} path if its last vertex is a sink and an \emph{open}
path otherwise. 

We will denote by $p \circ v$ the path $p$,
to which we add the extra vertex $v$ at the end.
We write $|p|$ for the length of path $p$.
A path $q$ is said to the a \textit{descendant} of $p$ if it contains $p$ as a prefix. In that case we call $p$ an \textit{ancestor} of $q$. We also say that $p \circ v$ is a \textit{child} of $p$ (who is then a \textit{parent} of $p \circ v$). Within a given set of paths $P$, we denote the sets of children, parent, ancestors, and descendants of a path $p$ by $C_P(p)$, $P_P(p)$, $A_P(p)$, and $D_P(p)$ respectively.
By $D_P(p, \ell)$ we denote the descendants
$q\in D_P(p)$ with $|q| = |p| + \ell$. Furthermore, we write $D_P(p, \le\ell) = \bigcup_{\ell'\le\ell} D_P(p, \ell')$.
The set of paths that ends at a vertex $v$ is denoted by $I_P(v)$.
If the set of paths $P$ is clear from the context, we may choose to omit the subscript for
convenience.

The conditions for a set of paths $Q$ to form a degree-$k$ solution are the following:
\begin{enumerate}
    \item $(s)\in Q$ for every source $s$,
    \item $|I_Q(v)| \le 1$ for every $v\in V$, and\label{en:cong1}
    \item $|C_Q(p)| = k$ for every path $p\in Q$ that is open.
 \end{enumerate}
   
We will also consider a relaxed version of~\eqref{en:cong1}, where we allow higher values than $1$.
Then we call maximum over all $|I_Q(v)|$ the congestion
of the solution. The motivation for looking at
solutions with (low) congestion is that we can remove
any congestion by reducing $k$ by the same factor.
\begin{lemma}\label{lem:remove-cong}
  Let $Q$ be a degree-$k$ solution with congestion $K$.
  Then in polynomial time we can compute a degree-$k / K$ solution without congestion.
\end{lemma}
\begin{proof}
  Consider a bipartite multigraph $(U\cup U', F)$ on two copies $U, U'$ of the vertices in $V$ that appear in $Q$.
  The graph has an edge $(u, v)\in F$ for every
  path $p\in Q$ that ends in $u$ and for which
  $p \circ v \in C(p)$.
  Then the degree of a vertex $u'\in U'$ is exactly the congestion of this vertex. Similarly, the degree of a vertex $u\in U$ is $k$ times its congestion.
  We consider a fractional selection of edges $x_e = 1/K$ for each $e\in F$.
  Here, each vertex $u\in U$ has $\sum_{e\in\delta(u)} x_e \ge k/K$ and each
  vertex $u'\in U'$ has
  $\sum_{e\in\delta(u')} x_e \le 1$, where
  $\delta(u)$ are the edges incident to $u$.
  As explained at the beginning of the section, we assume that $K$ divides $k$ and therefore $K/k$ is integer.
  By integrality of the bipartite matching polytope
  there exists also an integral vector $x'$ that satisfies these bounds. This corresponds to a degree-$k / K$ solution without congestion.
\end{proof}

\subsection{Bounded depth solution}
The following proof follows exactly the arguments
of a similar statement in~\cite{bateni2009maxmin}. We repeat it here for convenience.
\begin{lemma}\label{lem:bounded-depth}
  Let $Q$ be a degree-$k$ solution. Then there exists
  a degree-$k/2$ solution $Q'\subseteq Q$ where $|p|\le \log_2 n$
  for every $p\in Q'$.
\end{lemma}
\begin{proof}
  We iteratively derive $Q'$ from $Q$.
  For every $d = 1,2,\dotsc,n$ we consider the paths
  $p\in Q$ with $|p| = d$.
  For $d = 1$ we add all these paths to $Q'$.
  Then given the paths of length $d$ in $Q'$ we
  select for each of them the subset of $k/2$
  children in $Q$, which has the least number of
  descendants (assuming for simplicity that $2$ divides $k$).
  For every $d$ we will now bound the total number of
  descendants in $Q$ of paths of this length, namely
  \begin{equation*}
      n_d = \sum_{p\in Q' : |p| = d} |D_Q(p)| \ .
  \end{equation*}
  Notice that the descendants are counted in set $Q$ and not $Q'$.
  Clearly, we have $n_1 \le n$, since $|Q|\le n$.
  Then since we remove
  the children with the largest number of descendants,
  we get $n_{i+1} \le n_i / 2$ for every $i$.
  Thus, $n_{\log_2 n} = 0$.
\end{proof}

\subsection{Local congestion and layered instances}
A crucial concept in our algorithm are solutions
with locally low congestion. It will later be shown
that such a solution suffices to derive a solution
with (globally) low congestion.
\begin{definition}{\rm
  Let $Q$ be a solution and $\ell\in\mathbb N$.
  We say that $Q$ has an \emph{$\ell$-local congestion of
  $L$}, if for every $p\in Q\cup\{\emptyset\}$ and $v\in V$ we have
  \begin{equation*}
      |I_{D(p, \le\ell)}(v)| \le L \ .
  \end{equation*}}
\end{definition}
For sake of clarity, let us note the special case of $p = \emptyset$ in the definition above.
In this case, $D(p, \le\ell)$ is simply the set
of all paths in $P$ with length at most $\ell$ (potentially starting at different sources).
Throughout the paper we use values of the order $\ell = O(\log\log n)$ and
$L = \mathrm{poly}(\log\log n)$.
The usefulness of local congestion
is captured by the following lemma, which we will
prove in Section~\ref{sec:loc-to-glob}.
\begin{lemma}\label{lem:loc-to-glob}
  Let $Q$ be a degree-$k$ solution with $\ell$-local congestion of $L$ and global congestion $K$,
  where $\ell \ge \log_2 K$ and $K \ge \log_2 n$.
  Then we can compute in polynomial time a degree-$k/(8\ell)$ solution $Q'\subseteq Q$, which has global congestion at most
  \begin{equation*}
      O(\ell^7 L) \ .
  \end{equation*}
\end{lemma}
It remains to show how to compute a solution with
low local congestion.
The abstraction of local congestion and the
lemma on bounded depth allows us to reduce at a low
expense to instances in layered graphs.
\begin{definition}{\rm
A layered instance has layers $L_0\dot\cup L_1\dot\cup\dotsc\dot\cup L_{h} = V$
such that $L_0$ consists of all sources
and edges go only from one layer $L_i$ to the next
layer $L_{i+1}$.
}
\end{definition}

\begin{lemma}\label{lem:layerreduction}
  In polynomial time we can construct a layered instance
  with $h = \log_2 n$ such that if there exists a degree-$k$
  solution for the original instance, then there exists
  a degree-$k/2$ solution for the layered instance.
  Further, any degree-$k'$ solution with $\ell$-local
  congestion $L$ and global congestion $K$ in the layered instance can in
  polynomial time be transformed to a degree-$k'$
  solution for the original instance with $\ell$-local congestion $\ell L$ and global congestion $K\log_2 n$.
\end{lemma}
\begin{proof}
  Let $L_0$ be the set of sources. Then for each $i=1,2,\dotsc,\log n$ let $L_i$ be a copy of
  of all vertices $V$.
  We introduce an edge from $u\in L_i$ to $v\in L_{i+1}$
  if $(u, v)$ is an edge in the original instance.
  The new set of sinks is the union of all sink vertices in all copies.
  
  Consider now a degree-$k$ solution $Q$ for the original instance. By Lemma~\ref{lem:bounded-depth} there exists
  a degree-$k/2$ solution $Q'$ where each $p\in Q'$
  has $|p| \le \log_2 n$. For each such
  $p = (v_1,v_2,\dotsc,v_t)$
  we introduce a path $p' = (v'_1,v'_2,\dotsc,v'_t)$ in the layered instance where $v'_i$ is the copy of $v_i$ in $L_i$.
  
  Now let $Q'$ be a degree-$k'$ solution with $\ell$-local
  congestion $L$ and global congestion $C$ in the layered
  instance. We transform $Q'$ to a solution $Q$ for the
  original instance by replacing each path
  $p' = (v'_1, v'_2, \dotsc, v'_t)$ by a path
  $p = (v_1, v_2,\dotsc, v_t)$, where $v'_i$ is
  a copy of $v_i$. Since there are only $\log_2 n$
  copies of each vertex, the global congestion increases
  by at most a factor of $\log_2 n$. For the local
  congestion consider a path $p\in Q$.
  This path was derived from a path $p'\in Q'$.
  Notice that any path $q'\in D(p', \le \ell)$
  ends in some vertex in
  $L_{|p|+1},L_{|p|+2},\dotsc,L_{|p|+\ell}$. Thus, there are only $\ell$ copies of each
  vertex that $q'$ can end in.
  Consequently, the $\ell$-local congestion can increase
  at most by a factor of $\ell$.
\end{proof}
\begin{lemma} \label{lem:locallygood}
    Let $K = 2^{11}\log^3 n$, $\ell = 10\log\log n$, and $L = 2^{10} \ell^2$.
    Given a layered instance with a single source and optimum $k$, we can in quasi-polynomial time compute a degree-$k/(64\ell)$ solution with 
    $\ell$-local congestion at most $L$ and
    global congestion at most $K$.
\end{lemma}
This lemma is proven in Section~\ref{sec:locallygood}. The lemmas above would allow us already to obtain
our main result for instances with a single source. To generalize to multiple sources we present a black-box
reduction on layered instances. This is proved in Section~\ref{sec:localsearch}.
\begin{lemma}\label{lem:localsearch}
Suppose we have a quasi-polynomial time $\alpha$-approximation for the max-min degree
bounded arborescence problem on layered instances with a single
source.
Then there is also a quasi-polynomial time $256\alpha$-approximation for layered graphs and an arbitrary number of sources.
\end{lemma}

\subsection{Connecting the dots}
\begin{theorem}
\label{thm:main}
  In quasi-polynomial time we can compute a $\mathrm{poly}(\log\log n)$-approximation
  for the max-min degree arborescence problem. 
\end{theorem}
\begin{proof}
  First we prove the theorem on layered instances with a single source.
  Let $k$ be the optimum of the given instance.
  Using Lemma~\ref{lem:locallygood} we can find
  a degree-$\Omega(k/\ell)$ solution with $\ell$-local congestion $L$ and global congestion $K$. Here
  $K = O(\log^3 n)$, $\ell = O(\log\log n)$ with $2^\ell \ge K$, and $L = O(\ell^2)$.
  Next, we apply Lemma~\ref{lem:loc-to-glob} to turn this into a degree-$\Omega(k/\ell^2)$ solution with global congestion at most $O(\ell^7 L) = O(\ell^9)$. Using Lemma~\ref{lem:remove-cong} we can convert
  this to a degree-$\Omega(k/\ell^{11})$ solution without congestion. We therefore have
  an $O(\ell^{11})$-approximation algorithm
  for a single source on layered graphs and Lemma~\ref{lem:localsearch} implies that we can extend this
  to an arbitrary number of sources.

  We now turn our attention to instances that are not necessarily layered.
  Let again $k$ be the optimum. Using Lemma~\ref{lem:layerreduction} we construct a layered
  instance that is guaranteed to contain a degree-$k/2$ solution.
  Thus, with our algorithm for layered instances we can obtain a degree-$\Omega(k/\ell^{11})$ solution for it.
  This solution has $\ell$-local congestion at most $1$ and global congestion
  at most $1$. Using Lemma~\ref{lem:layerreduction} we can construct a degree-$\Omega(k/\ell^{11})$
  solution for the original (non-layered) instance with $\ell$-local congestion at most $\ell$ and global
  congestion at most $\log n$.
  Using again Lemma~\ref{lem:loc-to-glob} we obtain a degree-$\Omega(k/\ell^{12})$ solution with global congestion
  at most $O(\ell^8)$. Finally, applying Lemma~\ref{lem:remove-cong} we obtain a degree-$\Omega(k/\ell^{20})$
  solution without congestion. In particular, our approximation ratio is
  \begin{equation*}
      O(\ell^{20}) = \mathrm{poly}(\log\log n) \ . \qedhere
  \end{equation*}
\end{proof}
\subsection{Bottom-to-top pruning}
In this subsection we will describe a method of pruning that is used in
the proofs of Lemmas~\ref{lem:loc-to-glob} and~\ref{lem:locallygood}.

Suppose we are given a degree-$k$ solution
and we want to remove some of the paths, for example,
because they cause high congestion.
This may lead to some other paths having few children
and consequently these need to be removed as well.
In general, even a very small amount of removals
can lead to the whole solution getting corrupted, that
is, ultimately sources may need to be removed as well.
Below, we present a condition that guarantees
that certain parts of the solution remain intact. For a better intuition on this lemma and the bottom-to-top pruning technique, we refer the reader to Appendix~\ref{sec:bottom-to-top_pruning}.
\begin{lemma}\label{lem:bot-to-top}
Let $Q$ be a degree-$k$ solution.
Let $R\subseteq Q$ be a set of paths
that is supposed to be removed.
Let $\ell \ge 2$ such
that for every $p \in Q\setminus R$ we have
that at most $k^{\ell} / (8\ell)^2$ many descendants $q\in D(p, \ell)$
with $q\in R$.
Then we can compute in polynomial time a solution
$Q' \subseteq Q \setminus R$ such that
\begin{enumerate}
    \item $Q'$ is a degree-$k/2\ell$ solution and
    \item we have $(s)\in Q'$ for every source $s$, such
    that for any distance $\ell' \le \ell$
    there are at most $k^{\ell'} / (8\ell)$ many $q\in D((s), \ell')$ with $q\in R$.\label{en:survive}
\end{enumerate}
\end{lemma}
\begin{proof}
We assume that no path
$p\in R$ has a descendant also in $R$.
This is without loss of generality, since removing the
former implies that the latter
will be removed, and omitting the latter from $R$ still
keeps the premise of the lemma valid.
In particular, this assumption allows us to assert
that also paths in $R$ satisfy the bound on
the number of descendants in $R$.

We prune the solution from longest paths to shortest
paths:
We remove a path if it is in $R$ or
if more than $(1 - 1/2\ell)k$ many of its children
were removed.
Then we prove a stronger variant of (\ref{en:survive}) inductively, namely, that any path of length
$1$, $1+\ell$, $1+2\ell$, etc.\ satisfies the implication
(or an ancestor of it is removed).
Let $p$ be a path with $|p| = 1 + t\cdot \ell$
that satisfies the premise of
(\ref{en:survive}), but is not necessarily a singleton.
Further, assume that all paths of length
$1 + (t + 1)\ell$ satisfy the implication of~(\ref{en:survive}).
Let $\ell'\le \ell$.
Each of the distance-$\ell'$ descendants of $p$
has at most
$(8\ell)^{-2} k^{\ell}$ many distance-$\ell$ descendants
in $R$.
Consequently, $p$ has at most
\begin{equation*}
    k^{\ell'} \cdot (8\ell)^{-2} k^{\ell}
\end{equation*}
distance $(\ell + \ell')$-descendants in $R$.
Thus, at most $(8\ell)^{-1} k^{\ell}$ many of $p$'s
distance-$\ell$ descendants have more than
$(8\ell)^{-1} k^{\ell'}$ distance-$\ell'$
descendants belonging to $R$. Summing over all values
of $\ell'$, we have that at most $1/8 \cdot k^{\ell}$ many distance-$\ell$ descendants of $p$
do not satisfy the premise of~(\ref{en:survive}).
Next, let us show in a second induction that for every $\ell' \le \ell$, of the distance-$\ell'$ descendants of $p$ at most
\begin{equation*}
     \frac 1 8 \left(1 + \frac{2}{\ell} \right)^{\ell - \ell' + 1} k^{\ell'}
\end{equation*}
many are removed. For the base case we sum
the distance-$\ell$ descendants that do not satisfy
the premise of~(\ref{en:survive}) and the descendants
that are
themselves in $R$, which together are at most
\begin{equation*}
    \frac 1 8 k^{\ell} + \frac{1}{(8\ell)^2} k^{\ell} \le \frac 1 8 \left(1 + \frac{2}{\ell} \right)^{\ell - \ell + 1} k^{\ell} \ .
\end{equation*}
Now assume that we removed at most $1/8\cdot (1 + 2/\ell)^{\ell-\ell'} k^{\ell'+1}$ paths from
the distance-$(\ell'+1)$ descendants.
For each distance-$\ell'$ descendants that we remove
because of few remaining children, there are
$(1 - 1/(2\ell)) k$ many distance-$(\ell'+1)$ descendants
that we removed. The bound from this 
and the number of distance-$\ell'$ descendants in $R$
lets us bound the number of distance-$\ell'$ descendants
that we remove by
\begin{equation*}
     \frac 1 k \left(1 - \frac{1}{2\ell}\right)^{-1} \cdot \frac 1 8 \left(1 + \frac{2}{\ell}\right)^{\ell-\ell'} k^{\ell'+1} + \frac{1}{8\ell} k^{\ell'}
     \le \frac 1 8 \left(1 + \frac{2}{\ell} \right)^{\ell - \ell' + 1} k^{\ell'} \ .
\end{equation*}
It follows with $\ell'=1$ that we remove at most
$1/8 \cdot (1 + 2/\ell)^{\ell} k \le e^{2}/8 \cdot k < (1 - 1/(2\ell)) k$
children of $p$. Hence, $p$ is not removed itself
by the procedure.
\end{proof}

\subsection{Probabilistic lemmas}
We recall here two probabilistic results that we use extensively.


\begin{lemma}[Chernoff's bound, see \cite{chernoff1952measure}]
\label{lem:Chernoff}
Let $X_1,\ldots, X_n$ be independent random variables that take value in $[0,a]$ for some fixed $a$. Let $S_n=\sum_{i=1}^n X_i$. Then we have, for any $\delta \ge 0$,
\begin{equation*}
    \mathbb P[S_n\ge (1+\delta)\mathbb E[S_n]]\le \exp\left(-\frac{\delta^2 \mathbb E[S_n]}{(2+\delta)a} \right).
\end{equation*}
\end{lemma}

\begin{lemma}[Constructive Lov\'asz Local Lemma, see \cite{moser2010constructive}]
\label{lem:LLL}
Let $\mathcal X$ be a finite set of mutually independent random variables in a probability space. Let $\mathcal A$ be a finite set of events determined by these variables. For any $A\in \mathcal A$, let $\Gamma_{\mathcal A}(A)$ be the set of events $B\in \mathcal A$ such that $A$ and $B$ depend on at least one common variable. If there exists an assignment of reals $x : \mathcal A \mapsto (0, 1)$
such that for all $A\in \mathcal A$,
\begin{equation*}
    \mathbb P[A] \le  x(A) \cdot \prod_{B\in \Gamma_{\mathcal A}(A)}(1-x(B)) \ ,
\end{equation*}
then there exists an assignment of values to the variables $\mathcal X$ not triggering any of the events in $\mathcal A$. Moreover, there exists a randomized algorithm that finds such an assignment in expected time
\begin{equation*}
    |\mathcal X| \cdot \sum_{A\in \mathcal A} \frac{x(A)}{1-x(A)} \ .
\end{equation*}
\end{lemma}

\section{Local to global congestion}\label{sec:loc-to-glob}
This section is to prove Lemma~\ref{lem:loc-to-glob}.
Let $Q$ be a degree-$k$ solution with
$\ell$-local congestion $L$ and global congestion $K$.
We partition $Q$ by length: let $Q_i$ be the set
of paths $p \in Q$ with $|p| = i$.
Further, we split the paths into $\ell$
groups $G_1,G_2,\dotsc,G_{\ell}$, where
$G_j = Q_j \cup Q_{j + \ell} \cup Q_{j + 2\ell} \cup \cdots$.
For some $p\in G_j$ we write $G(p) = G_j$.
Roughly speaking, we proceed as follows. We sample from
each $G_{\ell}$ half of the paths and throw away all others (including their descendants).
Then we move to $G_{\ell-1}$ and do the same. We
continue until $G_1$ and then repeat the same a second time,
stopping afterwards.
The sampling is done in a way that guarantees that each
path retains a quarter of its children at the end.
We prove with Lov\'asz Local Lemma that in each step we can reduce the congestion significantly.
Since it seems unclear how to argue directly about the worst case congestion,
we will argue about the congestion aggregated over
many paths, which we will formalize next.

Consider a path $p \in Q$ and its close descendants in $D(p, \le \ell)$. Recall, that $D(p, \le \ell)$ contains all descendants of length at most $|p| + \ell$.
Intuitively, if many of the direct children of $p$ have high congestion (more precisely, the vertex that they end in), this is bad for $p$ as well: if they have high congestion, we may not be able
to keep many of these children for $p$, which means we might not be able to include $p$ itself
in the solution.
Let $\mathrm{cong}_G(p)$ be the congestion
of the last vertex in $p$, but restricted to paths in the group $G(p)$. In other words,
\begin{equation*}
    \mathrm{cong}_G(p) = |\{q\in G(p) \mid q \text{ ends in the same vertex as } p\}| \ .
\end{equation*}
The restriction to other paths in $G(p)$ is only for technical reasons and almost at no cost: if we can achieve
that every vertex is used by only few paths in each $G_j$ (i.e., $\mathrm{cong}_{G}(p)$ is small for all $p\in Q$), the overall congestion can only be worse by a factor $\ell$.
An important quantity in the following will be the total congestion of descendants $D(p, \ell')$ of $p$
at some distance $\ell' < \ell$, that is,
\begin{equation}
   \mathrm{cong}_{G}(D(p, \ell')) = \sum_{q\in D(p, \ell')} \mathrm{cong}_{G}(q) \ , \label{eq:ell-cong}
\end{equation}
Since during the procedure the number of children may differ between
groups $G_j$, we will use $k(G_j)$ to describe the current number of children
for every open path in $G_j$.
Our intermediate goal will be to bound the totals~\eqref{eq:ell-cong}
for some $p\in G_j$ in terms of
$k(p, \ell') = \prod_{j'=j}^{j+\ell'-1} k(G_j)$ (an upper bound on $|D(p, \ell')|$).

Notice that intially~\eqref{eq:ell-cong} can be at most $K \cdot k(p, \ell')$.
When sampling down the paths in $G_{j}$,
we want to show that this
reduces~(\ref{eq:ell-cong}) significantly for paths $p\in G_j$ and all
$\ell'<\ell$. This is captured in the following lemma.
\begin{lemma}\label{lem:sample-reduce}
  Assume we are given a solution that has an $\ell$-local congestion of at most $L$ and global congestion of at most $K$, where $\ell \ge \log K$ and $K\ge \log n$.
  From the paths in $G_j$ form pairs where each pair shares
  the same parent.
  Then select i.i.d.\ one path from each pair and remove it
  and all its descendants.
  Let $\mathrm{cong}_{G}$ and $\mathrm{cong}_{G}'$ be the congestion count before and after the removal and similarly $k$ and $k'$ the children count.
  Then we have with positive probability
  for every remaining $p\in G_j$ and $\ell' < \ell$ that
      \begin{equation*}
        \frac{\mathrm{cong}_{G}'(D(p, \ell'))}{k'(p, \ell')}
        \le c \ell^4 L +  \frac 1 2 \left(1 + \frac{1}{\ell}\right) \frac{\mathrm{cong}_{G}(D(p, \ell'))}{k(p, \ell')} \ ,
      \end{equation*}
  where $c$ is a fixed constant.
  Furthermore, we can obtain such a sampling in expected
  polynomial time.
\end{lemma}

\begin{proof}
We can rewrite
\begin{equation*}
    \mathrm{cong}_{G}(D(p, \ell')) = \sum_{q\in G(p)} \mathrm{cong}_G(D(p, \ell'), D(q, \ell')) \ ,
\end{equation*}
where $\mathrm{cong}_G(D(p, \ell'), D(q, \ell'))$ is the number of pairs $p' \in D(p, \ell'),q'\in D(q, \ell')$
that end in the same vertex.
We remove every term in the sum with probability $1/2$, so in expectation the sum will reduce by
$1/2$. Furthermore, each term $\mathrm{cong}(D(p, \ell'), D(q, \ell'))$
is bounded by $L \cdot |D(p, \ell')| \le L \cdot k(p, \ell')$, because we have low local congestion. This will give us good concentration. Notice also that $k(p,\ell') = k'(p,\ell')$.

We now group the terms in the sum by their size.
For $p\in G_j$ let $G(p,\ell',t)$ be the set of
$q\in G(p) = G_j$ with $\mathrm{cong}_G(D(p, \ell'), D(q, \ell')) \in [L \cdot k(p, \ell') \cdot 2^{-(t+1)}, L \cdot k(p, \ell') \cdot 2^{-t})$. Let $Q'$ be the set of paths remaining after sampling down
(without taking into account those that are removed recursively).
Let $c$ be a large constant to be specified later.
Depending on whether $t$ is small or large, we define bad events $\mathcal B(p,\ell',t)$ for each $p\in G_j$ as
\begin{align*}
    \sum_{q\in G(p,\ell',t) \cap Q'} \hspace{-1.5em} &\mathrm{cong}_G(D(p, \ell'), D(q, \ell')) \\
    &> 24 c \ell^3 L \cdot k(p, \ell') + \frac 1 2 \left(1 + \frac{1}{\ell}\right) \sum_{q\in G(p,\ell',t)} \hspace{-1em} \mathrm{cong}_G(D(p, \ell'), D(q, \ell')) &\text{ if } t \le 2\ell,  \\
\sum_{q\in G(p,\ell',t) \cap Q'} \hspace{-1.5em} &\mathrm{cong}_G(D(p, \ell'), D(q, \ell')) \\
    &> 24 c \ell^3 L \cdot \frac{1}{K} k(p, \ell') + \frac 1 2 \left(1 + \frac{1}{\ell}\right) \sum_{q\in G(p,\ell',t)} \hspace{-1em} \mathrm{cong}_G(D(p, \ell'), D(q, \ell')) &\text{ if } t > 2\ell.
\end{align*}
Since from the experiment the total congestion cannot
increase, we have a probability of $0$ for
all bad events where $\mathrm{cong}_{G(p,\ell',t)}(D(p, \ell')) :=$
\begin{equation*}
    \sum_{q\in G(p,\ell',t)} \hspace{-1em} \mathrm{cong}_G(D(p, \ell'), D(q, \ell'))
    \le \begin{cases}
    12 c \ell^3 L \cdot k(p, \ell') &\text{ if } t\le 2\ell, \\
    12 c \ell^3 L \cdot \frac{1}{K} k(p, \ell') &\text{ otherwise.}
    \end{cases}
\end{equation*}
For the remaining bad events, we will now derive an upper bound on the probabilities. This holds trivially also
for the zero probability events.
From Chernoff's bound we get
\begin{align*}
    \mathbb P[\mathcal B(p,\ell',t)] &\le
    \mathrm{exp}\left( - \frac{(1/\ell^2) \cdot 1/2 \cdot \mathrm{cong}_{G(p,\ell',t)}(D(p, \ell'))}{(2 + 1/\ell) \cdot 2^{-t} L \cdot k(p, \ell')}\right) \\
    &\le \mathrm{exp}\left( - \frac{\mathrm{cong}_{G(p,\ell',t)}(D(p, \ell'))}{6\ell^2 \cdot 2^{-t} L \cdot k(p, \ell')}\right) \\
    &\le \begin{cases}
    \exp\left(- 2 c \ell 2^t \right) &\text{ if } t\le 2\ell, \\
    \exp\left(- 2 c \ell 2^t/ K \right) \le n^{-10} &\text{ otherwise.}
    \end{cases}
\end{align*}
Towards applying LLL, we set the values of the bad events as
\begin{equation*}
    x(\mathcal B(p, \ell', t)) = \begin{cases}
    \exp\left(- c\ell 2^t \right) &\text{ if } t\le 2\ell, \\
    n^{-5} &\text{ otherwise.}
    \end{cases}
\end{equation*}
The experiment involves binary variables $V$.
An event $\mathcal B(p,\ell',t)$ depends on at most $K \cdot 2^{t+1}$ many variables.
A variable $V$ influences at most $2 K \ell \cdot 2^{t+1}$ type-$t$ bad events.
Notice that $2K\ell \cdot 2^{t + 1} \le \exp(c \ell 2^t) / (4 K^2)$ for $c$ sufficiently large (since $\ell \ge \log K$).
Let $\Gamma_t(V)$ be the set of
all these events for a specific variable $V$ and a specific
value of $t$. Then
\begin{equation*}
    \prod_{B\in\Gamma_t(V)} (1 - x(B))
    \ge \left(1 - \exp\left(- c \ell 2^t\right)\right)^{2 K\ell 2^{t+1}}
    \ge 1 - \frac{1}{4 K^2} \ .
\end{equation*}
Notice that $t$ can range only from $0$ to $\log n$.
Furthermore, since every vertex has congestion at most $K$,
the number of paths in $Q$ is at most $K n$.
Together, we can upper bound
the total number of events by $K n\ell\log n\le n^5$.
Thus, for $t \le 2\ell$ and $c$ sufficiently large it holds that
\begin{align*}
    \mathbb P[\mathcal B(p,\ell',t)] &\le x(\mathcal B(p, \ell', t)) \cdot \exp(-c \ell 2^t) \\
    &\le x(\mathcal B(p, \ell', t)) \cdot \left(1 - \frac 1 e\right)^{4 \ell \cdot 2^t} \\
    &\le x(\mathcal B(p, \ell', t)) \cdot \prod_{t' = 0}^{\ell} \left(1 - \frac{1}{4 K^2} \right)^{K \cdot 2^{t+1}} \cdot \left(1 - \frac{1}{n^5}\right)^{n^5} \\
    &\le x(\mathcal B(p,\ell',t)) \cdot \prod_{B\in \Gamma(\mathcal B(p,\ell',t))} (1 - x(B)) \ .
\end{align*}
Similarly, for $t > 2\ell$ we have
\begin{align*}
    \mathbb P[\mathcal B(p,\ell',t)] &\le x(\mathcal B(p, \ell', t) \cdot \exp(-c \ell 2^t / K) \\
    &\le x(\mathcal B(p, \ell', t)) \cdot \left(1 - \frac 1 e\right)^{4 \ell \cdot 2^t / K} \\
    &\le x(\mathcal B(p, \ell', t)) \cdot \prod_{t' = 0}^{\ell} \left(1 - \frac{1}{4 K^2} \right)^{K \cdot 2^{t+1}} \cdot \left(1 - \frac{1}{n^5}\right)^{n^5} \\
    &\le x(\mathcal B(p,\ell',t)) \cdot \prod_{B\in \Gamma(\mathcal B(p,\ell',t))} (1 - x(B)) \ .
\end{align*}
Hence, by LLL we have with positive probability
that none of the bad events occur.
If none of them occur, then by summing up bounds fixed $p$ and $\ell'$ we get
\begin{multline*}
    \mathrm{cong}'(D(p, \ell')) \le 2\ell \cdot 24 c \ell^3 L \cdot k(p, \ell') + \log n \cdot 24 c \ell^3 L \cdot k(p, \ell') / K + \frac 1 2 \left(1 + \frac{1}{\ell}\right) \mathrm{cong}(D(p, \ell')) \\
    \le 100 c \ell^4 L \cdot k(p, \ell') + \frac 1 2 \left(1 + \frac{1}{\ell}\right) \mathrm{cong}(D(p, \ell')) \qedhere
\end{multline*}
\end{proof}
\begin{lemma}\label{lem:sample-no-increase}
  Consider a successful run of the random experiment in Lemma~\ref{lem:sample-reduce} where
  we sample down the paths in $G_j$ and satisfy the stated inequalities.
  For each path $p$ (potentially not in $G_j$) and every $\ell'\in\{1,2,\dotsc,\ell\}$
  we have
  \begin{equation*}
   \frac{\mathrm{cong}_G'(D(p, \ell'))}{k'(p,\ell')} \le c\ell^4 L + \left(1 + \frac{1}{\ell} \right) \frac{\mathrm{cong}_G(D(p, \ell'))}{k(p,\ell')} \ .
  \end{equation*}
  Here $c$ is the constant from Lemma~\ref{lem:sample-reduce}.
\end{lemma}
Lemma~\ref{lem:sample-reduce} only shows that the average congestion reduces for descendants
of paths in the group $G_j$, where we sample down.
Conversely, Lemma~\ref{lem:sample-no-increase} says that for all other
groups it does not increase significantly.
\begin{proof}
  Let $G_{j'} = G(p)$.
  We can assume without loss of generality that $j' < j \le j' + \ell'$ (modulo $\ell$),
  since the congestion can only decrease and $k(p, \ell')$ in the other case
  would not change. Let $\ell'' = j - j'$ (modulo $\ell$).
  Further, let $D'(p, \ell')$ be the distance-$\ell'$ descendants of $p$ after sampling down $G_j$ and $D(p, \ell')$ before it.
  Then $D'(p, \ell'')$ contains half of the elements $D(p, \ell'')$.
  Thus,
  \begin{align*}
      \frac{\mathrm{cong}_G'(p, \ell')}{k'(p,\ell')} &= \frac{1}{k'(p, \ell')}\sum_{q\in D'(p, \ell'')} \mathrm{cong}_G'(q, \ell' - \ell'') \\
      &= \frac{2}{k(p, \ell')}\sum_{q\in D'(p, \ell'')} \mathrm{cong}'_G(q, \ell' - \ell'') \\
      &\le \frac{2}{k(p, \ell')} \sum_{q\in D'(p, \ell'')} [c \ell^4 L + \frac 1 2 \left(1 + \frac{1}{\ell}\right) \mathrm{cong}_G(q, \ell' - \ell'')] \\
      &\le c\ell^4 L + \left(1 + \frac{1}{\ell}\right) \mathrm{cong}_G(p, \ell') \ . \qedhere
  \end{align*}
\end{proof}

\begin{lemma}
  Given a degree-$k$ solution $Q$ with $\ell$-local congestion $L$ and global congestion $K$,
  we can compute a degree-$k/4$ solution $Q'\subseteq Q$ with
  \begin{equation*}
      \mathrm{cong}_G(p, \ell) \le 3 c \ell^4 L \left(\frac k 4\right)^{\ell} + \left(1 + \frac{1}{e}\right)^2\frac{K}{2^\ell}
  \end{equation*}
  for all $p \in Q'$.
  Here $c$ is the constant from Lemma~\ref{lem:sample-reduce}.
\end{lemma}
\begin{proof}
  We will perform the sampling from Lemma~\ref{lem:sample-reduce} for $G_{\ell}$,$G_{\ell-1}$,
  \dots,$G_{1}$ and then again the same a second time.
  The reason is that we want that for every group $G_j$ that $G_{j},G_{j-1},\dotsc,G_{j-\ell}$ (index modulo $\ell$)
  are down-sampled at least once in this order.
  
  Let $G_j = G(p)$ and consider the first
  time that we sample down $G_{j + \ell-1}$ (index taken modulo $\ell$).
  Let $D(p, \ell')$ be the distance-$\ell'$ descendants of $p$
  before this sampling. Then for each $q\in D(p, \ell - 1)$ we have
  \begin{equation*}
      \frac{\mathrm{cong}_G(D(q, 1))}{k(q, 1)} \le K \ .
  \end{equation*}
  This is due to the fact that sampling down cannot increase the congestion
  on any vertex and initially all vertices have congestion at most $K$.
  After sampling according to Lemma~\ref{lem:sample-reduce} we have
  \begin{equation*}
      \frac{\mathrm{cong}_G(q, 1)}{k(q, 1)} \le c \ell^4 L + \left(1 + \frac{1}{\ell}\right)\frac{K}{2}  \ .
  \end{equation*}
  In the next step we are sampling down $G_{j + \ell - 2}$. We have for each $q \in D(p, \ell - 2)$ that
  \begin{equation*}
      \frac{\mathrm{cong}_G(D(q, 2))}{k(q, 2))} = \frac{1}{k(q, 1)} \sum_{q'\in C(q)} \frac{\mathrm{cong}_G(q', 1)}{k(q', 1)} \le c\ell^4 L + \left(1 + \frac{1}{\ell}\right) \frac{K}{2} \ .
  \end{equation*}
  Thus, after sampling
  \begin{equation*}
      \frac{\mathrm{cong}_G(D(q, 2))}{k(q, 2)} \le c\ell^4 L + \frac 1 2 \left(1 + \frac{1}{\ell}\right) c \ell^4 L + \left(1 + \frac{1}{\ell}\right)^2 \frac{K}{4} \le 2c\ell^4 L + \left(1 + \frac{1}{\ell}\right)^2 \frac{K}{4} \ .
  \end{equation*}
  Continuing this argument, after we sample down $G_{j}$ we have
  \begin{equation*}
      \frac{\mathrm{cong}_G(p, \ell)}{k(p, \ell)} \le 3 c\ell^4 L + \left(1 + \frac{1}{\ell}\right)^{\ell} \frac{K}{2^\ell} \ .
  \end{equation*}
  After $G_j$ there may be at most $\ell$ more steps of sampling down, after which we finally have
  \begin{equation*}
      \frac{\mathrm{cong}_G(p, \ell)}{k(p, \ell)} \le 3c\ell^4 L + \left(1 + \frac{1}{\ell}\right)^{2\ell} \frac{K}{2^\ell} \le 3c\ell^4 L + \left(1 + \frac{1}{e}\right)^{2} \frac{K}{2^\ell} \ . \qedhere
  \end{equation*}
\end{proof}
We will now conclude the proof of Lemma~\ref{lem:loc-to-glob}.
Using the previous lemma, we obtain a degree-$k/4$ solution $Q'$.
Let $A := 3c\ell^4 L + (1 + 1/e)^2 K / 2^{\ell}$ be the upper bound on
average distance-$\ell$ congestion. Assuming without loss of generality that
$c$ is sufficiently large, we have that $A \le 3c\ell^4 L$.
Let $R$ be the set of all paths $p$ with $\mathrm{cong}_G(p) > 16A\ell^2$.
We remove these paths using Lemma~\ref{lem:bot-to-top}.
We use the property that
each path has at most a $(8\ell)^{-2} (k/4)^{\ell}$ many 
of its ancestors at distance $\ell$ in $R$.
This follows directly from the bounded average congestion.
The lemma implies that we can remove those high
congestion paths and still keep a solution
where each remaining path has $k/(4 \cdot 2\ell) = k / 8\ell$ children.
Since we start with a $\ell$-local congestion of at most $L\le 16A\ell^2$ and this cannot
be increased by only removing paths, we have that none of the paths of length at most
$\ell$ are removed and thus all sources satisfy the premise
of~(\ref{en:survive}) of Lemma~\ref{lem:bot-to-top} and remain in the solution.
Indeed, the value of $\mathrm{cong}_G(p)$ is now bounded by $16A\ell^2$
for all remaining paths. We recall that the actual congestion is at most
a factor $\ell$ higher than $\max_{p} \mathrm{cong}_G(p)$, that is,
\begin{equation*}
    16A\ell^3 \le 64c \ell^7 L \ .
\end{equation*}
\section{Computing a solution with locally low congestion}\label{sec:locallygood}
The goal of this section is to prove Lemma~\ref{lem:locallygood}. We recall that we are in a layered graph with vertices partitioned into layers $L_0, L_1,\dotsc, L_{h}$ where $h = \log n$ and a single source $s$. The source $s$ belong to layer $L_0$ and edges can only go from vertices in some layer $L_i$ to
the next layer $L_{i+1}$. 

In the following we will extensively argue about
paths that start in the source. For the
remainder of the section every path $p$ that
we consider is implicitly assumed to start at the source and then traverse (a prefix of)
the layers one by one.
Slightly abusing notation, we sometimes
use $L_i$ also to denote the set of paths
ending in a vertex of $L_i$, that is,
\begin{equation*}
    \bigcup_{v\in L_i} I(v) \ ,
\end{equation*} 
and $T$ to describe the set of closed paths (recall that those are the paths that end at a sink). Let $P$ refer to the set of all possible paths and notice that by virtue of the layers
we have that $|P| \le n^{h} \le n^{\log n+1}$.
We will now describe
a linear programming relaxation, which goes back to
Bateni et al.~\cite{bateni2009maxmin}. The intuition behind the linear program is to select paths similarly to the way we describe solutions,
see Section~\ref{sec:prelim}.
We have a variable $x(p)$ for each path $p$ that in an integral solution takes value $1$ if the path $p$ is contained in an arborescence and $0$ otherwise.
\begin{align}
    \sum_{q\in C(p)} x(q) &= k \cdot x(p) &\forall p\in P\setminus T \label{eqLP:flowLPdemand}\\
     \sum_{q\in I(v)\cap D(p)} x(q) &\leq x(p) &\forall p\in P,v\in V \label{eqLP:flowLPcapacity}\\
    x((s)) &= 1 & \label{eqLP:flowLPdemandroot}\\
    x &\ge 0 & 
\end{align}
Here we assume that $k$ is the highest value for which the
linear program is feasible, obtained using
a standard binary search framework. Moreover, we assume that $k\ge 2^{10}(\log \log n)^8$ in the rest of the section. This is at little cost, since a $1/k$-approximation is easy to obtain (see the proof of Theorem~\ref{thm:main}) and is already sufficient for our purposes.
The first two types of constraints describe that each open
path has many children and each vertex has low congestion
(in fact, no congestion).
Constraint~\eqref{eqLP:flowLPcapacity} comes from a lift-and-project idea.
For integral solutions it would be
implied by the other constraints, but
without it, there
could be situations with continuous variables, where for example we take a path $p$ with only value $1/k$ and then a single child of $q$ with value $1$. Such situations easily lead to large integrality gaps, which we can avoid by this constraint.

Since the graph has $h+1$ many layers, this linear program has $n^{O(h)}$ variables and constraints and therefore can be solved in time $n^{O(h)}$. We refer to this relaxation as the \textit{path LP}. In order to prove Lemma \ref{lem:locallygood}, we will design a rounding scheme. Before getting to the main part of the proof, we will first preprocess the fractional solution to sparsify its support.

\subsection{Preprocessing the LP solution}
Our first step is to sparsify the path LP solution $x$ to get another sparser solution (i.e., with a limited number of non-zero entries). For ease of notation, we might need to take several times a copy of the same path $p\in P$. We emphasize here that two copies of the same path are different objects. To make this clear, we will now have a \textbf{multiset} $P'$ of paths but we will slightly change the parent/child relationship between paths. Precisely, for any path $q'\in P'$ is assigned as child to a \textbf{unique} copy $p'\in P'$ such that $q$ was a child of $p$ in the set $P$. With this slight twist, all the ancestors/descendants relationships extend to multisets in the natural way. For instance, we will denote by $D_{P'}(p)$ the set of descendants of $p$ in the multiset $P'$. Again, the set of closed path in $P'$ will be denoted by $T'$ and $s$ refers to the source. We assume that there is a unique copy of the trivial path $(s)\in P'$.

In this step we will select paths such that each open path
has $k \log^2 n$ children instead of the $k$ children one would expect. However, we use a function $y(p)$ that assigns a
weight to each path and this weight decreases from layer
to layer, modelling that the children are actually picked fractionally each with a $1/4\log^2 n$ fraction of the weight of the parent. Thus, taking the weights into account we are actually picking only $k/4$ children for each path.

Formally, the preprocessing of the LP will allow us to obtain a multiset of path $P'$ such that 
\begin{align}
    \sum_{q\in C_{P'}(p)} y(q) &= \frac{k}{4} \cdot y(p) &\forall p\in P'\setminus T'\label{eqSparseLP:demand}\\
     \sum_{ q\in I_{P'}(v)\cap D_{P'}( p)} y( q) &\leq 2y( p) &\forall p\in P',v \in V\label{eqSparseLP:capacity}\\
     \text{where }y(p) &=  \frac{1}{(4\log^2 (n))^i}  & \forall i\leq h, \forall  p\in L_i\label{eqSparseLP:granularity}
\end{align}
We obtain such a solution $P'$ in a similar way as the randomized rounding in~\cite{bateni2009maxmin, chakrabarty2009allocating}, which achieves polylogarithmic congestion. 
The fact that we select more paths, but only fractionally
gives us better concentration bounds, which allows us
to lose only constant congestion here.
For completeness we give
the proof in Appendix~\ref{sec:Sparsify}.

\subsection{The main rounding}
We start this part with the sparse multiset of paths $P'$ with
the properties as above. The discount value $y$ can be thought of fractionality in the sense that each $p\in P'$ is taken to an extend of $y(p)$.
We will proceed to round this fractional solution to an integral solution $Q$ that is \textit{locally nearly good}, a concept we will define formally below. Intuitively, this specifies that
the number of paths that have locally high congestion can be removed without loosing much with Lemma~\ref{lem:bot-to-top}. 
We fix $\ell=10 \log \log n$ for the rest of this section.
Let $\textrm{cong}(v \mid Q)$ be the global congestion of vertex $v$ induced by $Q$, that is, the number of paths in $Q$ ending in $v$. 
For paths $p\in Q\cup\{\emptyset\}$ and $q\in D_Q(p)$ We denote by $\textrm{cong}_p(q \mid Q)$ the \textit{local congestion} induced on the endpoint of $q$ by descendants of $p$.
We consider all paths descendants of $\emptyset$.

A locally nearly good solution is a multiset of paths $Q\subseteq P'$ (where again every path has a \emph{unique} parent among the relevant copies of the same path) that has the following properties:
\begin{enumerate}
    \item one copy of the trivial path $(s)$ belongs to $Q$;
    \item every open path has $k/32$ children;
    \item no vertex has global congestion more than $2^{10} \log^3(n)$;
    \item\label{en:local-cong} for every $p\in Q \cup \{\emptyset\}$ and $\ell'\le \ell$ we have
    \begin{equation*}
        |\{q \in D_Q(p, \ell') \mid \mathrm{cong}_p(q \mid Q) > 2^{10} \ell^2\}| \le \frac{1}{\ell^2} \left(\frac{k}{32}\right)^{\ell'} \ .
    \end{equation*}
\end{enumerate}
From a locally nearly good solution we will then derive
a solution of low local congestion by removing all paths
with high local congestion using bottom to top pruning (Lemma~\ref{lem:bot-to-top}).
Condition~\ref{en:local-cong} is tailored to ensure that the number
of such high local congestion paths is small enough so that the
lemma succeeds.

To obtain such a nearly good solution, we will proceed layer by layer, where the top layers are already rounded integrally and the bottom layers are still fractional (as in the preprocessing). However, unlike the case of the preprocessing, we cannot argue with high probability and union bounds, since some properties we want only deviate from expectation by $\textrm{poly}(\log\log n)$ factors. To obtain a locally nearly good solution, we will use Lov\'asz Local Lemma (LLL) in every iteration, where one iteration rounds one more layer. In the following, we describe the rounding procedure, the bad events and an analysis of their dependencies and finally we apply LLL.

\paragraph{The randomized rounding procedure.} We proceed layer by layer to round the solution $P'$ (with discount $y$) to an integral solution $Q$. We start by adding a single copy the trivial path $(s)$ to the partial solution $Q^{(0)}$. Assume we rounded until layer $i$, that is, we selected the final multiset of paths $Q^{(i)}$ to be used from all paths in $L_{\le i}$.

To round until layer $i+1$ we proceed as follows. Every open path $p\in Q^{(i)}\cap L_i$ selects exactly $k/16$ children where the $i$-th child equals $q\in C_{P'}(p)$ with probability equal to 
\begin{equation*}
    \frac{1}{k\log^2 n} = \frac{y(q)}{\sum_{q'\in C_{P'}(p)}y(q')}\ .
\end{equation*}
The selection of each child is independent of the choices made for other children. We then let $Q^{(i+1)}$ be the union of $Q^{(i)}$ and all
newly selected paths.
The reason that each path selects $k/16$ children
instead of the $k/32$ many that were mentioned before is
that we will later lose half of the children.
We will repeat this procedure until reaching the last layer $h$ and we return $Q=Q^{(h)}$.

\paragraph{Definitions related to expected congestion.}
In order for this iterated rounding to succeed we need to
avoid that vertices get high congestion (in the local and the global sense). It is not enough to keep track only of the congestion of vertices
in the next layer that we are about to round, but we also need to maintain that
the expected congestion (over the remaining iterations) remains low on all vertices in later layers. Hence, we will define quantities that
help us keep track of them. To avoid confusion, we remark
that the quantities we
will define do not exactly correspond to the expected congestion of the vertices: notice that we sample less children for each path than $P'$ has.
Intuitively, $P'$ has $k/4$ children per open path (see Equation~\eqref{eqSparseLP:demand}), but we only sample $k/16$ many.
The quantities we define below would be the expectations if we would sample $k/4$ children instead. Roughly speaking, this gives us an advantage
of the form that the expectation always decreases by a factor $4$ when
we round one iteration.
Apart from the intuition on the expectation, the definitions also
incorporate a form of conditional congestion, which means that we
consider the (expected) congestion based on random choices made so far.
This is similar to notions that appear in Appendix~\ref{sec:Sparsify} for the preprocessing step.

First, we define the fractional congestion induced by some path $p\in P'$ on a vertex $v\in V$ as follows.
\begin{equation*}
      \textrm{cong}(p,v) := \sum_{q\in D_{P'}(p)\cap I_{P'}(v)} y(q).
\end{equation*}
Using this definition, we will define the conditional fractional congestion at step $i$ induced by $p$ on a descendant $q$ as follows (writing $v_q$ for the endpoint of $q$).
\begin{equation*}
    \textrm{cong}_p(q \mid Q^{(i)}) := \begin{cases}
   |\{q' \in D_{Q^{(i)}}(p)\cap I_{Q^{(i)}}(v_q) \}| & \textrm{if }
    q\in L_{\le i}\\
    \sum_{q'\in D_{Q^{(i)}}(p)\cap L_i} \frac{1}{y(q')} \textrm{cong}(q',v_q)  & \textrm{otherwise.}
    \end{cases}
\end{equation*}
We will use this definition only for $p\in Q^{(i)}$. If $v_q$ belongs to one of the first $i$ layers (i.e., $v_q$ is in the integral part corresponding to the partial solution), this is simply the number of descendants of $p\in Q^{(i)}$ in our partial solution that end at $v_q$. Otherwise, this is the total fractional congestion induced by all the paths that we actually selected in our partial solution $q\in Q^{(i)}\cap L_i$ and that are descendants of $p$. The name \textit{conditional} comes from the term $1 / y(q')$ which is simply the fact that we condition on having already selected the path $q'$ once integrally.
The global congestion at step $i$ on a vertex $v\in V$ is defined similarly with
\begin{equation*}
    \textrm{cong}(v\mid Q^{(i)}) := \begin{cases}
     |\{q \in Q^{(i)}\cap I(v)\}| & \textrm{if }
    v\in L_{\le i}\\
    \sum_{q\in Q^{(i)}\cap L_i} \frac{1}{y(q)} \textrm{cong}(q,v)  & \textrm{otherwise.}
    \end{cases}
\end{equation*}

\paragraph{The first type of bad event: global congestion.} The naive way to bound the global congestion would be to simply to bound the global congestion on each vertex and make a bad event from exceeding this. However, to manage dependencies between bad events and get better concentration bounds, we partition the set of paths according to how much fractional load their ancestors are expected to put on $v$ and then bound the congestion incurred by each group on $v$. More precisely, assume we rounded until layer $i$ and are now trying to round until layer $L_{i+1}$. Consider any vertex $v\in L_{\ge i+1}$. The decisions made in this step
of the rounding concern which paths in $L_{i+1}$ will be selected.
Fix an integer $t\ge 0$ and and let $P^{(i+1)}_{v,t}$ to be the set of all paths $p \in L_{i+1}$ such that 
\begin{equation*}
    \frac{\textrm{cong}(p,v)}{y(p)} \in (2^{-t},2^{-(t-1)}]\ .
\end{equation*}

We define the bad event $\mathcal B_1(v,t)$ as the event that
\begin{equation*}
    \sum_{p\in Q^{(i+1)} \cap P^{(i+1)}_{v,t}} \hspace{-1em} \frac{\mathrm{cong}(p, v)}{y(p)}  > 2 \mathbb E\bigg[\sum_{p\in Q^{(i+1)} \cap P^{(i+1)}_{v,t}} \hspace{-1em}\frac{\mathrm{cong}(p, v)}{y(p)}\bigg] + 2^{10} \log n \ .
\end{equation*}
Since $P^{(i+1)}_{v,t}$ partitions the paths in $L_{i+1}$, the bad events
for all $t$ together will bound the increase in the congestion on $v$:
this is because $\mathrm{cong}(v \mid Q^{(i+1)}) = \sum_{p\in Q^{(i+1)}\cap L_{i+1}} \mathrm{cong}(p, v) / y(p)$. We argue this formally in Lemma \ref{lem:iterlowcong} below.

At this point, it is worthwhile to mention that $t$ can only take values in $\{0,1,\dotsc,\log^2 n\}$.
Indeed, by Constraint~\eqref{eqSparseLP:capacity} we have that
\begin{equation*}
    \frac{\textrm{cong}(p,v)}{y(p)} \le 2
\end{equation*}
is ensured for all $p\in P'$. Moreover, we notice that $y$ has a low granularity (i.e. $y$ satisfies Equation \eqref{eqSparseLP:granularity}) which implies that either $ \textrm{cong}(p,v)=0$ or
\begin{equation*}
 \frac{\textrm{cong}(p,v)}{y(p)} \ge \frac{1}{(4\log^2 n)^{h-|p|}}\ge \frac{1}{(4\log^2 n)^{h}}\ .
\end{equation*}
Hence bad events $\mathcal B_1(v,t)$ will be instantiated only for $t= 0,1,\dotsc,\log^2 n$ (assuming here that $n$ is sufficiently large).

\begin{lemma}\label{lem:iterlowcong}
  Assuming no bad event $\mathcal B_1(v, t)$ has occurred in any iteration up
  to $i$, we have for every $v$ that
  \begin{equation*}
      \mathrm{cong}(v \mid Q^{(i)}) \le 2^{11} \log^3 n \ .
  \end{equation*}
\end{lemma}
\begin{proof}
  We argue inductively. For $i = 0$ we have
  \begin{equation*}
      \mathrm{cong}(v \mid Q^{(i)}) = \sum_{p\in Q^{(i)}\cap L_i} \frac{\mathrm{cong}(q, v)}{y(q)} = \sum_{s\in S} \mathrm{cong}((s), v) = \sum_{p\in I_{P'}(v)} y(v) \le 2 \ . 
  \end{equation*}
  Now assume that for $i \ge 0$ we have
  \begin{equation*}
      \mathrm{cong}(v \mid Q^{(i)}) \le 2^{11} \log^3 n \ .
  \end{equation*}
  Since the congestion of any $v\in L_{\le i}$ cannot change anymore,
  we may assume w.l.o.g.\ that $v\in L_{\ge i+1}$.
  For each $t$ let
  \begin{equation*}
      \mu_t = \mathbb E\bigg[\sum_{p\in Q^{(i+1)} \cap P^{(i+1)}_{v,t}} \hspace{-1em}\frac{\mathrm{cong}(p, v)}{y(p)}\bigg] \ .
  \end{equation*}
  Then $\sum_t \mu_t = 1/4 \cdot \mathrm{cong}(v \mid Q^{(i)}) \le 2^{9} \log^3 n$. To see this, denote by $N^{(i)}_q$ the number of copies of a path $q\in P'$ that is selected in $Q^{(i)}$. As a shorthand, we write $P(q)$ for the parent of $q$ in $Q^{(i+1)}$. By linearity of expectation we can write
  \begin{align*}
      \sum_t \mu_t &= \sum_t \ \mathbb  E\bigg[\sum_{p\in Q^{(i+1)} \cap P^{(i+1)}_{v,t}} \hspace{-1em}\frac{\mathrm{cong}(p, v)}{y(p)}\bigg]\\
      &=  \sum_t \ \mathbb  E\bigg[\sum_{p\in P' \cap P^{(i+1)}_{v,t}} \hspace{-1em}N^{(i+1)}_p\cdot \frac{\mathrm{cong}(p, v)}{y(p)}\bigg]\\
      &=  \sum_t  \sum_{p\in P' \cap P^{(i+1)}_{v,t}} \hspace{-1em}\mathbb E[N^{(i+1)}_p]\cdot \frac{\mathrm{cong}(p, v)}{y(p)}\\
      &=  \sum_t  \sum_{p\in P' \cap P^{(i+1)}_{v,t}} \  \frac{(k/16)\cdot y(p)}{(k/4)\cdot y(P(p))}\cdot \frac{\mathrm{cong}(p, v)}{y(p)}\\
      &= \frac{1}{4}\cdot \bigg[\sum_t  \sum_{p\in P' \cap P^{(i+1)}_{v,t}} \ \frac{\mathrm{cong}(p, v)}{y(P(p))}\bigg]\\
      &= \frac{1}{4}\cdot \bigg[\sum_t  \sum_{p\in P' \cap P^{(i)}_{v,t}} \ \frac{\mathrm{cong}(p, v)}{y(p)}\bigg] =  (1/4) \cdot \mathrm{cong}(v \mid Q^{(i)})\le 2^{9} \log^3 n\ .
  \end{align*}
  It follows that
  \begin{align*}
      \mathrm{cong}(v \mid Q^{(i+1)}) &= \sum_{p\in Q^{(i+1)}} \frac{\mathrm{cong}(p,v)}{y(p)} \\
      &= \sum_{t} \sum_{p\in Q^{(i+1)} \cap P^{(i+1)}_{v,t}} \frac{\mathrm{cong}(p,v)}{y(p)} \\
      &\le \sum_{t} [2 \mu_t + 2^{10} \log n] \\
      &\le 2^{10} \log^3 n + 2^{10} \log^3 n \\
      &\le 2^{11} \log^3 n \ . \qedhere
  \end{align*}
\end{proof}

\paragraph{The second type of bad event: Local congestion.}
Recall that we want to bound the congestion induced by
close descendants of some open path $p$.
Let $\ell'\le \ell$ and $p\in Q^{(i)}\cap L_{i'}$, where $i - \ell' \le i'\le i$. We define
\begin{align*}
    N^{(i+1)}(p, \ell') = |\{q\in D_{P'}(p,\ell') \mid \ &\mathrm{cong}_{p}(q \mid Q^{(i)}) \le 2^{10} \ell^2 \\ &\text{and } \mathrm{cong}_{p}(q \mid Q^{(i+1)}) > 2^{10} \ell^2\}| \ .
\end{align*}
These are the number of vertices with newly high local congestion (counting only the congestion induced by descendants of $p$). Moreover, note that by Constraint \eqref{eqSparseLP:demand} and \eqref{eqSparseLP:granularity}, we have that $|D_{P'}(p,\ell')| \le (k\log^2 n)^{\ell'}$.

Then, for any $p\in L_{i'}\cap Q^{(i)}$ where $i - \ell - 1 \le i' \le i - 1$, and $t\ge 0$, we define the set of \emph{marked} children of $p$ at step $i+1$ as the set $M^{(i+1)}(p)$ of all $q \in C_{Q^{(i)}}(p)$ such that
\begin{equation*}
    N^{(i+1)}(q, \ell') > \frac{1}{\log^{10} n} \left(\frac{k}{32}\right)^{\ell'}
\end{equation*}
for at least one $\ell' \le \ell$. Notice that the children of $p$ (in $Q^{(i)}$) are
already determined because $p$ is in $L_{\le i - 1}$.
We are now ready to state the second type of bad event. We define the bad event $\mathcal B_2(p)$ for any $p\in L_{i'}$ where $i - \ell -1 \le i' \le i - 1$, as the event that 
\begin{equation*}
    |M^{(i+1)}(p)|\ge \frac{1}{\ell^3} \cdot \frac{k}{16} \ ,
\end{equation*}
which means that a lot of children selected by $p$ become marked at step $i+1$. As we will ensure that this bad event never happens, we can guarantee that a large fraction of children $q$ of each path $p$ always satisfies that $N^{(i+1)}(q, \ell')$ is low for all $\ell'\le \ell$.
As we explain later, we will simply remove all other children and this way indirectly
prevent the existence of too many locally congested descendants for all the remaining paths. The advantage
of this indirect method is that we can apply good concentration bounds
on the children and thereby amplify the probability (indeed, note that the event that a child $q$ of $p$ is marked is independent of the same events for other children of $p$).

\paragraph{The third type of bad event: Keeping the source path $(s)$.} Because our algorithm will delete in the end all the paths that have been marked in some round, we need to ensure that the source path $(s)$ is never marked. Hence for the first $\ell$ steps of rounding, we define a single bad event $\mathcal B_3 (s)$ that is that the source path $(s)$ is marked during this step.

\paragraph{The probabilities of the bad events.} In this part, we derive upper bounds for the probability of each bad event.
Some bounds are sub-optimal to simplify later formulas.
\begin{lemma}
For any $v\in V$ and integer $t$ we have that
\begin{equation*}
    \mathbb P[\mathcal B_1(v,t)]\leq n^{-10 \cdot 2^{t-1}} \ .
\end{equation*}
\end{lemma}
\begin{proof} We need to prove that
\begin{equation*}
    \mathbb P\Bigg[\sum_{p\in Q^{(i+1)} \cap P^{(i+1)}_{v,t}} \hspace{-1em} \frac{\mathrm{cong}(p, v)}{y(p)}  > 2 \mathbb E\bigg[\sum_{p\in Q^{(i+1)} \cap P^{(i+1)}_{v,t}} \hspace{-1em}\frac{\mathrm{cong}(p, v)}{y(p)}\bigg] + 2^{10} \log n\Bigg]
    \le n^{-10 \cdot 2^{t-1}} \ .
\end{equation*}
Notice that $\mathrm{cong}(p, v) / y(p)$ is upper bounded by $2^{-(t-1)}$ by definition of $P^{(i+1)}_{v,t}$ and the sum can be
rewritten as a sum over the paths chosen in iteration $i+1$, where a path $p$
contributes $\mathrm{cong}(p, v) / y(p)$ to the sum.
These paths are chosen independently and
thus by a standard Chernoff bound we obtain that the probability is at most 
\begin{equation*}
    \exp\left( \frac{-2^{10} \log n}{2 \cdot 2^{-(t-1)}}\right)\le n^{-10 \cdot 2^{t-1}} \ . \qedhere
\end{equation*}
\end{proof}

\begin{lemma}\label{lem:problocal}
Let $\ell' \le \ell$ and $p\in Q^{(i)}\cap L_{i'}$, where $i-\ell'\le i'\le i$. Then
\begin{align}
    \mathbb P[N^{(i+1)}(p, \ell') > (1/\log^{10} n) \cdot (k / 32)^{\ell'}] &\le (\log n)^{-40\ell}\ .
\end{align}
Moreover,
\begin{align}
    \mathbb P[\mathcal B_3(s)] &\le (\log n)^{-1}\ .
\end{align}
\end{lemma}
\begin{proof} Recall that we are bounding the number of descendants of $p$ that
had low local congestion of at most $2^{10} \ell^2$ before, but now
get high local congestion.
More concretely,
consider some $q\in D_{P'}(p, \ell')$ with $\mathrm{cong}_{p}(q \mid Q^{(i)}) \le 2^{10}\ell^2$.
Then
\begin{equation*}
    \mathbb E\bigg[\mathrm{cong}_{p}(q \mid Q^{(i+1)})\bigg] = \frac{1}{4} \cdot \mathrm{cong}_{p}(q \mid Q^{(i)}) \le 2^8 \ell^2 \ .
\end{equation*}
We do not go into detail for this calculation here but similar calculations have been derived already in Lemma \ref{lem:iterlowcong} or Appendix \ref{sec:Sparsify}. This is exactly the same type or argument. 

Further, $\mathrm{cong}_{p}(q \mid Q^{(i+1)})$
can be decomposed into a sum of
independent random variables bounded by $2$. Indeed, recalling the definition of $\mathrm{cong}_{p}(q \mid Q^{(i+1)})$ we see that this quantity can be written as a sum over $1/y(q') \cdot \mathrm{cong}_{q'}(q \mid Q^{(i+1)})$, where a paths $q'$ are taken from a certain multiset of chosen paths. Constraint \eqref{eqSparseLP:capacity} of the path LP ensures that this is always at most $2$. Those variables are independent because of the randomized rounding that selects children independently of each other.
Hence, by a Chernoff bound we have
\begin{equation*}
    \mathbb P\bigg[\mathrm{cong}_{p}(q \mid Q^{(i+1)}) > 2^{10} \ell^2\bigg] \le \exp(-2^7\ell^2) \ .
\end{equation*}
Therefore, by linearity of expectation we obtain
\begin{equation*}
    \mathbb E[N^{(i+1)}(p, \ell')] \le \frac{(k\log^2 n)^{\ell'}}{\exp(2^7\ell^2)} \le \frac{1}{(\log n)^{2^7 \cdot 10\ell}} \left(\frac{k}{32}\right)^{\ell'} \ .
\end{equation*}
From Markov's inequality it follows that
\begin{equation*}
    \mathbb P[N^{(i+1)}(p, \ell') > (1/\log^{10} n) \cdot (k / 32)^{\ell'}] \le (\log n)^{-40\ell}\ .
\end{equation*}
For the second claim, notice that $\mathcal B_3(s)$ is simply the event that the source path $(s)$ is marked, which is equivalent to $N^{(i+1)}((s), \ell') > (1/\log^{10} n) \cdot (k / 32)^{\ell'}$ for at least one $\ell'\le \ell$. By a simple union bound, we get the desired result. 
\end{proof}

Finally, we upper-bound the probability of $\mathcal B_3(p)$.

\begin{lemma}
Let $p \in L_{i'}$ where $i - \ell - 1 \le i' \le i - 1$.
Then we have that
\begin{equation*}
    \mathbb P[\mathcal B_3(p)]\le \exp\left(-\sqrt{k} \right) \ .
\end{equation*}
\end{lemma}
\begin{proof} To prove this, note that a child $q$ of $p$ (in the set $ Q^{(i)}$) is marked independently of other children of $p$. This is because $q$ being marked depends only on the random choices made by descendants of $q$, which are independent of descendants of other children of $p$. Therefore using Lemma \ref{lem:problocal} and a standard union bound we obtain that each child of $p$ is marked independently with probability at most $\ell \cdot (\log n)^{-40\ell}$. 
Recall that $k\ge 2^{10} (\log\log n)^8\ge (32\ell^3)^2$.
We note that $p$ has $k/16$ children in $Q^{(i)}$ and therefore the probability that more than $(1/\ell^3)\cdot (k/16)$ of them are marked is at most
\begin{equation*}
    \exp\left(- \frac{k}{32\ell^3}\right)\le  \exp\left(- \sqrt{k}\right)\ . \qedhere
\end{equation*}
\end{proof}

\paragraph{The dependencies of the bad events.}
For any bad event $\mathcal B$ we define $\Gamma_1(\mathcal B)$ to be the set of bad events of the first type $\mathcal B_1(v,t)$ that depend on the bad event $\mathcal B$. Similarly, let $\Gamma_{2}(\mathcal B)$ the set of bad events $\mathcal B_2(p)$ that depend on $\mathcal B$.
We will now upper bound the cardinality of these sets. Note that there is a single bad event of third type $\mathcal B_3(s)$ hence it is clear that $\Gamma_3(\mathcal B)\le 1$ for any bad event $\mathcal B$.
For the rest, we remark that the focus here is on simplicity of the terms rather
than optimizing the precise bounds.
\begin{lemma}
For any bad event $\mathcal B$, we have that
\begin{equation*}
    |\Gamma_1(\mathcal B)|\leq n^2 \ .
\end{equation*}
\end{lemma}
\begin{proof}
The statement holds trivially, since there are at most $n\cdot (1 + \log^2 n)\le n^2$ bad events of the first type in total ($n$ possibilities of the choice of vertex $v$ and less than $n$ possibilities for the choice of $t$).
\end{proof}
Before proving the dependencies to events of type $2$, we will prove an auxiliary lemma that concerns the events affected
by the children picked by one particular path $p\in Q^{(i)}\cap L_i$.
\begin{lemma}\label{lem:p-influence}
 Let $p\in Q^{(i)}\cap L_i$. Then the choice of children picked by $p$ affects
  in total at most $\log n$ events of type $2$.
\end{lemma}
\begin{proof} In order to influence a bad event $B_2(q)$ for some $q$, it must be that $p$ is a descendant of $q$ in the multiset $Q^{(i)}$. Any path has at most $h\le \log n$ ancestors. 
\end{proof}

\begin{lemma}
Assuming that no bad event has happened until the current iteration of rounding,
for any $v\in V$ and integers $t\ge 0$, we have 
\begin{align}
    |\Gamma_{2}(\mathcal B_1(v,t))| \leq 2^t \cdot (\log n)^{10}\ .
\end{align}
\end{lemma}
\begin{proof} We first notice that the bad event $\mathcal B_1(v,t)$ depends only on the random choices made by paths $q\in Q^{(i)}\cap L_i$ that have at least one child $q' \in C_{P'}(q)$ such that 
\begin{equation*}
    \frac{\textrm{cong}(q',v)}{y(q')} \ge 2^{-t}.
\end{equation*}
Let us count how many such paths $q$ can exist in $Q^{(i)}$. For each such path we have that 
\begin{align*}
    \frac{\textrm{cong}(q,v)}{y(q)}= \sum_{q'\in C_{P'}(q)} \frac{\textrm{cong}(q',v)}{y(q)}
    &= \sum_{q'\in C_{P'}(q)} \frac{\textrm{cong}(q',v)}{(4\log^2 n)y(q')}\\
    &\ge \frac{ 2^{-t}}{4\log^2 n} \ ,
\end{align*}
where we used here the granularity of $y$, see Constraint \eqref{eqSparseLP:granularity}. Since we assume that no bad event happened  so far, we can apply Lemma~\ref{lem:iterlowcong} to get
\begin{equation*}
    \sum_{q\in Q^{(i)}\cap L_i} \frac{\textrm{cong}(q,v)}{y(q)} = \textrm{cong}(v\mid Q^{(i)}) \le 2^{11} \log^3 n \ .
\end{equation*}
Hence, there can be at most $2^t\cdot  2^{13} \log^5 n$ such paths $q$. 
We conclude with Lemma~\ref{lem:p-influence} to bound the number of bad events of second type influenced by each $q$
and obtain 
\begin{equation*}
    |\Gamma_{2}(\mathcal B_1(v,t))| \le 2^t \cdot (2^{13}\log^5 n)\cdot (\log n) \le 2^t \cdot (\log^{10} n) \ . \qedhere
\end{equation*}
\end{proof}
\begin{lemma}
For any $p\in L_{i'}$ where $i - \ell - 1 \le i' \le i - 1$ we have
\begin{equation}
    |\Gamma_{2}(\mathcal B_2(p))|\leq k^\ell\cdot (\log n)\ .
\end{equation}
Furthermore,
\begin{equation}
    |\Gamma_{2}(\mathcal B_3(s))|\leq k^\ell\cdot (\log n)\ .
\end{equation}
\end{lemma}
\begin{proof} The bad event $\mathcal B_2(p)$ (or $\mathcal B_3(s)$) is entirely determined by the random choices made by the descendants of $p$ (or $(s)$) in $Q^{(i)}$ that are located in layer $L_i$. There are at most $(k/16)^\ell\le k^\ell$ such descendants $q$. By Lemma \ref{lem:p-influence}, each descendant can influence at most $\log n$ other bad events of the second type. This concludes the proof.
\end{proof}

\paragraph{Instantiating LLL.}
For the bad events we set
\begin{align}
    x(\mathcal B_1(v,t)) &= n^{-2} \ , \\
    x(\mathcal B_{2}(p)) &= \exp\left(-k^{1/3} \right)\\
    x(\mathcal B_{3}(s)) &= 1/2 \ .
\end{align}
Consider now a bad event $\mathcal B_1(v, t)$ and recall that $k\ge 2^{10}(\log \log n)^8$. Then
\begin{align*}
    \mathbb P[\mathcal B_1(v, t)] &\le n^{-10 \cdot 2^{(t-1)}} \\
    &\le n^{-5 \cdot 2^{t}}\\
    &\le n^{-2} \cdot (1 - n^{-2})^{n^2} \cdot \left(1-\exp(-k^{1/3})\right)^{2^t\cdot (\log n)^{10}} \cdot (1/2)\\
    &\le n^{-2} \cdot \prod_{\mathcal B \in \Gamma(\mathcal B_1(v, t))} (1-x(\mathcal B)) \ .
\end{align*}
Next, we get
\begin{align*}
    \mathbb P[\mathcal B_2(p)] &\le \exp\left(-\sqrt{k} \right) \\
    &\le \exp\left(-k^{1/3} \right) \cdot (1-n^{-2})^{n^2} \cdot \left(1-\exp(-k^{1/3})\right)^{k^\ell \cdot (\log n)} \cdot (1/2)\\
    &\le x(\mathcal B_2(p))\cdot  \prod_{\mathcal B\in \Gamma(\mathcal B_2(p))} (1 - x(B)) \ .
\end{align*}
Similarly, for $\mathcal B_{3}(s)$ we have
\begin{align*}
    \mathbb P[\mathcal B_3(s)] &\le (\log n)^{-1} \\
    &\le (1/2)\cdot (1-n^{-2})^{n^2} \cdot \left(1-\exp(-k^{1/3})\right)^{k^\ell \cdot (\log n)} \\
    &\le x(\mathcal B_3(s))\cdot  \prod_{\mathcal B\in \Gamma(\mathcal B_3(s))} (1 - x(B)) \ .
\end{align*}
Hence, with positive probability none of the bad events occur
and we can compute such a solution in expected quasi-polynomial time (polynomial in $P'$).

\paragraph{Finishing notes.} We proved that we can round in (expected) quasi-polynomial time from layer $0$ to layer $h$ without any bad event ever happening.
Next, observe that we get that any sampled path $p\in Q \cap L_i$ has at most 
\begin{equation*}
    \bigg|\bigcup_{i'=i+1}^{i+\ell+1} M^{(i')}(p)\bigg|\le \ell \cdot \frac{1}{\ell^3} \cdot \frac{k}{16} \le  \frac{k}{32}
\end{equation*}
marked children. We now delete from the solution all the marked paths
obtaining a solution $Q'$. Note that the source remains as it is never marked. This means that any sampled open path retains at least $k/32$ children. For any path $p$, we denote by $N(p,\ell')$ the number of its descendants at distance $\ell'$ that have congestion induced by $p$ bigger than $2^{10}\ell^2$. Then if $p$ was never marked at any round, we have that 
\begin{equation*}
    N(p,\ell')\le \frac{1}{\log n} \left(\frac{k}{32} \right)^{\ell'}\ .
\end{equation*}
Indeed, if $p$ belongs to layer $i$ then $\textrm{cong}_p(q\mid Q^{(i)})=1/y(p) \cdot \textrm{cong}(p,v_q)\le 2$ by Constraint~\eqref{eqSparseLP:capacity}. Hence the first time we instantiate bad events involving local congestion induced by the path $p$, it is the case that none of its descendants are bad. Then if $p$ is never marked we have that 
\begin{equation*}
    N^{(i')}(p,\ell')\le \frac{1}{\log^{10} n}\cdot \left(\frac{k}{32} \right)^{\ell'}
\end{equation*} at every iteration $i\le i'\le i+\ell$, which means that few of the good descendants become bad at each round. Hence we obtain that 
\begin{equation*}
    N(p,\ell')\le \sum_{i'=i}^{i+\ell} N^{(i')}(p,\ell') \le \frac{1}{\log n} \left(\frac{k}{32} \right)^{\ell'}\ .
\end{equation*}
This ensures that any remaining path has few bad descendants (in terms of local congestion).

We now define a set of paths $R\subseteq Q$, which contains all paths with high local congestion. More precisely, let $R$ contain for every path $p\in Q$ and $\ell'\le\ell$ all paths $q\in D_{Q}(p, \ell')$ such that $\mathrm{cong}_p(q \mid Q') > 2^{10} \ell^2$. Since we have no more marked paths,
each $p$ has at most $(\log n)^{-1}\cdot  (k/32)^{\ell'}\le 1/(8\ell)^2\cdot (k/32)^{\ell'}$ descendants at distance $\ell'$ in $R$.
Applying Lemma~\ref{lem:loc-to-glob} we can remove $R$ and obtain a solution
$Q' \subseteq Q \setminus R$,
where each open path has $k / (64\ell)$ children and the source remains in the solution. Moreover, the solution has no more paths
with $\ell$-local congestion more than $2^{10}\ell^2$ and the global congestion is at most $2^{11} \log^3 n$.
This finishes the proof of Lemma~\ref{lem:locallygood}.

\section{From single source to multiple sources}
\label{sec:localsearch}
We devise an algorithm that solves the problem with multiple
sources based on an algorithmic technique first introduced
by Haxell in the context of hypergraph matchings and then applied to a range of other problems, including the restricted
assignment version of the Santa Claus problem. Haxell's
algorithmic technique can be thought of as a (highly non-trivial) generalization
of the augmenting path algorithm for bipartite matching. 
Our algorithm makes only calls to the $\alpha$-approximation
for a single source. The algorithm itself requires quasi-polynomial time as well (even if the $\alpha$-approximation ran in polynomial time).
Our concrete variant relies on the following simple, but powerful subprocedure.
\begin{lemma}\label{lem:ls-pruning}
  Let $Q$ be a (single) degree-$k'$ arborescence in a layered graph. Let $R\subseteq Q$ be a set of paths where
  at most $(k'/4)^i$ many end in each layer~$i$ .
  Then
  there is a degree-$k'/4$ arborescence $Q'\subseteq Q\setminus R$. Furthermore, $Q'$ can be computed in polynomial time.
\end{lemma}
\begin{proof}
  We start pruning the arborescence from bottom to top and
  we maintain at every layer~$i$ that we remove at most $(k'/2)^i$ paths in it. At the last layer $i=h$ we remove
  simply the paths that are in $R$. These are by definition
  at most $(k'/4)^i \le (k'/2)^i$ many.
  Then assume we already pruned layers $i,i+1,\dotsc,h$
  and that we removed at most $(k'/2)^i$ many paths in layer~$i$. Now in layer $i-1$ we again remove the at most $(k'/4)^{i-1}$ many paths in $R$, but also those open paths
  where more than $3/4 \cdot k'$ many children were removed in layer~$i$. The number of open paths removed this way
  is at most $(k'/2)^i / (3/4 \cdot k') \le 2/3 \cdot (k'/2)^{i-1} < (k'/2)^{i-1}$. Clearly, this procedure maintains that
  every remaining open path has at least $k'/4$ children
  and the source never gets removed.
\end{proof}

Throughout the section we assume that $k$ is the optimum,
which is obtained through a binary search framework.
We now assume that we
have a degree-$k/256\alpha$ solution that already covers all
but one source $s_0\in S$. We will augment this solution
to one that covers all sources.
This is without loss of generality, since we can apply the
procedure iteratively $|S|$ times, adding one source at a time.

\paragraph*{Blocking trees and addable trees.}
Let $Q$ be our current solution (which does not cover $s_0$).
On an intuitive level, \emph{blocking trees} are $k/256\alpha$-degree arborescences
in our current solution $Q$, which we would like to remove from the solution. In order to remove them, we need to add other arborescences for their sources instead. The \emph{addable trees} are $k/32\alpha$-degree arborescences not in our
solution. Their sources are sources of blocking trees in $Q$
and we would like to add to them solution. The addable trees
in turn may be \emph{blocked} by blocking trees, which means
they overlap on some non-source vertex with a tree in $Q$,
preventing us from adding them to the solution.
We note that addable trees are (by a factor of $8$)
better arborescences than what we ultimately need.
As explained above, addable and blocking trees naturally form
an alternating structure, which is usually referred to as layers. In order to avoid conflicts with our other notion of layers, we call these \emph{rings}. The addable and blocking
trees in the $i$th ring will be denoted by $A_i$ and $B_i$.

\paragraph*{The algorithm.}
Initially we have no rings. We run our $\alpha$-approximation
to find a $k/\alpha$-degree arborescence for $s_0$ (without taking into account our current solution $Q$). We
reduce this to a $k/32\alpha$-degree arborescence and store
it as singleton set in $A_1$. Then we store in $B_1$ all
arborescences in $Q$ that intersect on any vertex with it.
If the total intersection on each layer $j$ is at most
$(k/128\alpha)^j$, then we can find through Lemma~\ref{lem:ls-pruning} a $k/128\alpha$-degree arborescence for $s_0$ which is disjoint
from $Q$. We reduce it to a $k/256\alpha$-degree arborescence and add it to the solution and terminate.
Otherwise we continue the algorithm, but now we indeed have
one ring of addable and blocking trees.

Assume now the algorithm in its current state has rings
$A_1, B_1,\dotsc,A_i,B_i$. Then we intialize $A_{i+1} = \emptyset$ and add addable trees to it in the following Greedy
manner. From our (layered) graph we produce a reduced instance
by removing any vertices
appearing in $A_1,B_1,\dotsc,A_i,B_i,A_{i+1}$ (except for the sources).
Then we iterate over any source $s$ that is either $s_0$ or
the source of a blocking tree in
$B_1\cup\cdots\cup B_i$. We call our $\alpha$-approximation
for $s$ on the reduced instance. If it outputs a
$k'$-degree solution with $k' \ge k/32\alpha$, we add it to $A_{i+1}$ (after reducing to degree exactly $k/32\alpha$).
We repeat this until we can no longer add any addable trees
to $A_{i+1}$. Then we construct $B_{i+1}$ by taking any arborescence in $Q$ that intersects with any addable tree
in $A_{i+1}$.
As before, we check if any addable tree in $A_{i+1}$ has a
total intersection of at most $(k/128\alpha)^j$ on each layer $j$
with solution $Q$. If so, we collapse:
we can compute through Lemma~\ref{lem:ls-pruning}
a $k/128\alpha$-degree arborescence
for the corresponding source $s$, which is disjoint from $Q$.
We reduce it to a $k/256\alpha$-degree arborescence and add
it to $Q$. Now $s$ has two arborescences in $Q$.
The other arborescence must have appeared as a blocking tree
in an earlier ring $B_{i'}$. We remove this blocking tree
from $Q$ and delete all rings after $i'$. Then we revisit
the addable trees in $A_{i'}$. Since we removed a blocking tree in $B_{i'}$,
one of them may now have a small intersection with all trees in $Q$ (as above). If so, we collapse this as well. We continue
so until no more collapse is possible, leaving us with a new solution $Q$ and a prefix of the previous rings.
We then continue the algorithm.
Once an addable tree for $s_0$ is added to $Q$
we terminate.

\paragraph*{Analysis.}
Our analysis relies on two lemmas.
\begin{lemma}\label{lem:ls-growthB}
  At any point of time in the local search algorithm
  we have for every ring that $|B_i| \ge |A_i|$.
\end{lemma}
\begin{proof}
  Recall that the addable trees in $A_i$ are vertex disjoint
  by construction. At the same time, none of them can be
  collapsed. This means for each of them there must be a layer
  $j$ such that the addable tree intersects with $Q$ (that is, with $B_i$) on
  $(k/128)^j$ many vertices. Let $R_j$ be the set of all
  intersecting vertices with $Q$ in layer~$j$ (across all addable trees in $A_i$). Then
  \begin{equation*}
      \sum_{j \ge 1} \frac{|R_j|}{(k/128)^j} > |A_i| \ .
  \end{equation*}
  On the other hand, each blocking tree in $B_i$ has at most $(k/256)^j$ many vertices in each layer~$j$. Thus,
  \begin{equation*}
        \sum_{j\ge 1} \frac{|R_j|}{(k/128)^j} \le \sum_{j \ge 1} |B_i| \cdot \frac{1}{2^j} \le |B_i| \ . \qedhere
  \end{equation*}
\end{proof}
\begin{lemma}
  After the Greedy procedure to construct $A_{i+1}$, we have
  that $|A_{i+1}| \ge \sum_{i' = 1}^i |B_{i'}|$.
\end{lemma}
\begin{proof}
  By Lemma~\ref{lem:ls-growthB} we have that
  \begin{equation*}
      \sum_{i'=1}^i |A_{i'}| + |B_{i'}| \le 2 \sum_{i'=1}^i |B_{i'}| \ .
  \end{equation*}
  Now assume toward constradiction that $|A_{i+1}| < \sum_{i'=1}^i |B_{i'}|$ and no more addable trees can be added. Recall that the reduced
  instance we consider removes all vertices appearing in
  $A_1,B_1,\dotsc,A_i,B_i,A_{i+1}$. By the bounds above we know that this removes a total of at most
  \begin{equation}
      (k/32)^j \cdot 3 \sum_{i'=1}^i |B_{i'}| \label{eq:ls-removedj}
  \end{equation}
  many vertices from each layer $j$.
  Now consider the optimal degree-$k$ solution restricted to the sources in $\bigcup_{i'=1}^i B_{i'}$.
  Each of the arborescences in this solution must overlap with the removed vertices on some layer $j$ on at least $(k/4)^j$ many vertices. Otherwise,
  by Lemma~\ref{lem:ls-growthB} there would be a degree-$k/4$ arborescence rooted in one of the sources of a blocking tree which is disjoint from the removed vertices. Consequently, our $\alpha$-approximation would have found a degree-$k/4\alpha$ arborescence that would have been added as an addable tree. Let $R_j$ be the intersection of the aforementioned optimal solution with the removed vertices.
  Then
  \begin{equation*}
      \sum_{j\ge 1} \frac{|R_j|}{(k/4)^j} \ge \sum_{i'=1}^i |B_{i'}| \ .
  \end{equation*}
  However, from~\eqref{eq:ls-removedj} it follows that
  \begin{equation*}
      \sum_{j\ge 1} \frac{|R_j|}{(k/4)^j} \le \left(3 \sum_{i'=1}^i |B_{i'}|\right) \cdot \sum_{j\ge 1} \left(\frac{1}{8}\right)^j < \sum_{i'=1}^i |B_{i'}| \ ,
  \end{equation*}
  a contradiction.
\end{proof}

From the lemmas above it follows that the local search
never gets stuck: at any time it can either collapse or
add a new (large) ring.
It remains to show that it terminates in quasi-polynomially
many steps. First we observe that the algorithm
never has more than $O(\log n)$ many rings.
By the previous two lemmas we have for every $i$ that
\begin{equation*}
    \sum_{i'=1}^{i+1} |B_{i'}| = |B_{i+1}| + \sum_{i'=1}^{i} |B_{i'}| \ge |A_{i+1}| + \sum_{i'=1}^{i} |B_{i'}| \ge 2 \sum_{i'=1}^{i} |B_{i'}| \ .
\end{equation*}
Consequently, the number of blocking trees grows exponentially
and the number of rings can be only logarithmic.
Now consider the potential
\begin{equation*}
    (|B_1|,|B_2|,\dotsc,|B_i|,\infty) \ ,
\end{equation*}
where $i$ is the number of rings.
This potential decreases lexicographically with every collapse or newly added ring. Number of possible potentials
is bounded by $n^{O(\log n)}$, which implies that the local
search terminates in quasi-polynomial time.

\bibliographystyle{alpha}
\bibliography{references}

\appendixpage
\section{Bottom-to-top pruning is not sufficient by itself}
\label{sec:bottom-to-top_pruning}
In this section, we make a short argument to illustrate why the bottom-to-top pruning seems not enough to fix the issues in previous works. Assume we built an arborescence of out-degree $k$ which is a complete $k$-ary tree where all the sinks are located at distance $h= \Theta(\log n/\log k)$ from the source (and all vertices in the tree have congestion 1 by definition). Then we let an adversary delete a $\beta$ fraction of the sinks and we try to maintain an solution of out-degree at least $k/\alpha$. For this, we need to delete all vertices that are left with less than $k/\alpha$ children.

The adversary can select the sinks to delete so that $\frac{\beta}{1-1/\alpha}$ fraction of the vertices at distance $h-1$ from the source keep an out-degree less than $k/\alpha$. We would then need to delete those. By the same principle, the adversary could have chosen the sinks to delete so that a $\frac{\beta}{(1-1/\alpha)^2}$ fraction of the vertices at distance $h-2$ keep less than $k/\alpha$ children (after we were forced to delete the bad vertices at distance $h-1$). More generally the adversary can choose the sinks to delete in order to force us to delete a $\frac{\beta}{(1-1/\alpha)^i}$ fraction of the vertices at distance $h-i$ from the source. In particular, the source looses a $\frac{\beta}{(1-1/\alpha)^{h-1}}\approx \beta \cdot \exp(h/\alpha)$ fraction of its children.

Let us plug in typical values now. We set $k=\log^{10} n$ so that $h=\Theta(\log n/\log \log n)$. Assume we wanted a sub-logarithmic approximation ratio so we were aiming for $\alpha=\log^{1-\epsilon} n$, for some fixed $\epsilon >0$. Then the source roughly looses $\beta \cdot \exp\left(\log^\epsilon n/\log \log n\right)\ge \beta \cdot \exp\left(\log^{\epsilon'} n\right)$ fraction of its children. Hence if we want the source to retain any good fraction of its children we can allow to delete only a $\beta=\exp\left(-\log^{\epsilon'} n \right)$  fraction of the sinks.

Note that if we were happy with $\alpha = h$, which means we keep only a $1/h$ fraction of the children for each vertex, then we can afford to delete any constant fraction of the sinks. But $h$ can be as big as $\Omega(\log n/\log \log n)$, and this is not good enough for our purposes. Note that our bottom-to-top pruning lemma (Lemma \ref{lem:bot-to-top}) circumvents this issue by only imposing a bounded number of deletions at distances $h'=\Theta(\log \log n)$, which gives much less layers for the deleted fraction to amplify.
\section{A challenging instance for randomized rounding}
\label{sec:hardinstance}
In this section, we discuss a specific limitation of the rounding algorithm in previous works (see \cite{chakrabarty2009allocating} and \cite{bateni2009maxmin}). Essentially, we explain why their techniques cannot give a better than $ \Theta (\frac{\log(n)}{\log \log(n)})$ approximation. This argument is based on an instance that we build below, which also illustrates why the bottom-to-top pruning cannot fix the issue by itself.
We reuse here the notation defined in Section \ref{sec:prelim}. For clarity, we restate here the LP relaxation of the problem (for a single source $s$), called the path LP.
\begin{align*}
    \sum_{q\in C(p)} x(q) &= k \cdot x(p) &\forall p\in P\setminus T \\
     \sum_{q\in I(v)\cap D(p)} x(q) &\leq x(p) &\forall p\in P,\forall v\in V \\
    x((s)) &= 1 &\forall s\in S \\
    x &\ge 0 & 
\end{align*}

Below we give a construction of a random instance inspired by \cite{halperin2007integrality}. The graph will have $h+1$ layers labeled from 0 to $h$ and a single source $s$ located in $L_0$. $n$ will be the number of vertices in the last layer $L_h$ and we choose $h=\Theta \left(\frac{\log n}{\log \log n}\right)$ so that $(\log^{10} n)^{h}=n$. The set of sinks corresponds exactly to the vertices in the last layer $L_h$. The graph until layer $h-1$ is a complete and regular tree where every internal vertex has exactly $\log^{20} n$ children. Vertices in layer $L_{h-1}$ will be called \textit{leaves}. Next, we describe how we connect layer $h-1$ to layer $h$.

Each vertex $v\in L_\ell$ (note that $v$ is a sink) runs the following sampling process, independently of other sinks.
Start at the source $s$. From the source select each vertex $u$ in $L_1$ independently at random with probability $1/\log^{10} n$. All the vertices in $L_2$ whose parent was selected is then selected independently at random with probability $1/\log^{10} n$. We repeat layer by layer, each vertex in layer $L_i$ whose parent was selected is then selected independently at random with probability $1/\log^{10} n$.

Finally, connect the sink $v\in L_h$ to all vertices in layer $L_{h-1}$ that were selected by the above random process initiated by $v$.

In the following we denote by $T_u$ for any $u\in \cup_{0\le i \le h-1}L_i$ the subtree rooted at vertex $u$ that is contained in layers $\cup_{0\le i \le h-1} L_i$. We first claim the following.
\begin{claim}
\label{cla:sink_degree}
With high probability, for any $j\le h-1$ no sink $v\in L_h$ is connected to more than 
\begin{equation*}
    \left(1+\frac{1}{\log(n)}\right)\cdot (\log^{10}(n))^{h-1-j}
\end{equation*} leaves in layer $\ell-1$ that belong to the subtree rooted at a vertex $u\in L_j$.
\end{claim}
\begin{proof}
To see this, note that during the random process initiated by the sink $v$, every internal node $u$ that was selected selects in expectation $\log^{20}(n)/\log^{10}(n)=\log^{10}(n)$ of its children. By a standard Chernoff bound the probability that this internal nodes selects more than 
\begin{equation*}
    \left(1+\frac{1}{\log^2 (n)}\right)\cdot \log^{10}(n)
\end{equation*} of its children is at most 
\begin{equation*}
    \exp\left(-\frac{1}{3\log^4 (n)}\cdot \log^{10} (n) \right)\le \exp\left(-\log^6(n)/3 \right).
\end{equation*}
Therefore with high probability, no internal node selects more than $\left(1+\frac{1}{\log^2 n}\right)\cdot \log^{10}(n)$ of its children which implies that the total number of selected nodes at layer $h-1$ that belong to the subtree $T_u$ (which are exactly the neighbors of $v$ in that subtree) is at most
\begin{equation*}
     \left(1+\frac{1}{\log^2 n}\right)^{h-1-j} (\log^{10} n)^{h-1-j}\le \left(1+\frac{1}{\log(n)}\right) (\log^{10} n)^{h-1-j}.
\end{equation*}
\end{proof}
Let $u$ be a vertex in layer $L_{h-1}$. Then $u$ is connected to each vertex $v\in L_h$ independently with probability $1/(\log^{10} n)^{h-1}$, therefore with high probability each vertex $u\in L_{h-1}$ has at least 
\begin{equation*}
    \frac{n}{\left(1+\frac{1}{\log(n)} \right)\cdot (\log^{10}(n))^{h-1}}=\frac{\log^{10} (n)}{1+\frac{1}{\log(n)} } 
\end{equation*} many neighbors in layer $h$.

With this observation and Claim \ref{cla:sink_degree}, it is easy to prove the following.
\begin{claim}
\label{cla:flowLPfeasible}
With high probability, the path LP is feasible on this instance for $$k=\frac{\log^{10}(n)}{1+\frac{1}{\log(n)}}\ .$$
\end{claim}
\begin{proof}
We set the following fractional values.
\begin{equation*}
    x(p)=
    \begin{cases}
    (1+1/\log n)^{-(h-1)}\cdot (\log^{10} n)^{-(h-1)} & \textrm{ if $p$ ends at layer $h$, and}\\
     (1+1/\log n)^{-i}\cdot (\log^{10} n)^{-i} & \textrm{ if $p$ ends at layer $i\le h-1$.}
    \end{cases}
\end{equation*}
Clearly, $x((s))=1$. We can easily verify the demand constraints: let $p$ be a path that ends at layer $i\leq h-2$, then
\begin{equation*}
     \sum_{q\in C(p)} x(q) =  \frac{\log^{20} n}{(1+1/\log n)^{i+1} \cdot (\log^{10} n)^{i+1}} =  \frac{k}{(1+1/\log n)^{i} \cdot (\log^{10} n)^{i}}= k \cdot x(p) \ .
\end{equation*}
If $p$ ends at layer $h-1$, then with high probability we have
\begin{equation*}
    \sum_{q\in C(p)} x(q) \ge \frac{\log^{10} (n)}{1+\frac{1}{\log n} }  \cdot  \frac{1}{(1+1/\log n)^{h-1}\cdot (\log^{10} n)^{h-1}} = k \cdot x(p) \ .
\end{equation*}
The capacity constraints are also easily verified with high probability.
Let $u$ be a vertex in layer $i\leq h-1$ and $p$ a path ending in layer $j\le i$ then
\begin{equation*}
    \sum_{q\in I(v)\cap D(p)} x(q) \le  \frac{1}{(1+1/\log n)^i\cdot (\log^{10} n)^{i}}\le  \frac{1}{(1+1/\log n)^j\cdot (\log^{10} n)^{j}}=x(p) \ .
\end{equation*}
If $v$ is a vertex in layer $h$ and $p$ a path ending in layer $i\le h-2$ then, using Claim \ref{cla:sink_degree}, we obtain
\begin{align*}
     \sum_{q\in I(v)\cap D(p)} x(q) &\leq \left(1+\frac{1}{\log n}\right)\cdot (\log^{10} n)^{h-1-i} \cdot  \frac{1}{(1+1/\log n)^{h-1}\cdot (\log^{10} n)^{h-1}}\\
     &\le \frac{1}{(1+1/\log n)^{h-2}\cdot (\log^{10} n)^{i}}\\
     &\le x(p) \ .
\end{align*}
If $v$ is a vertex in layer $h$ and $p$ a path ending in layer $i= h-1$ then
\begin{equation*}
     \sum_{q\in I(v)\cap D(p)} x(q) \le   \frac{1}{(1+1/\log n)^{h-1}\cdot (\log^{10} n)^{h-1}} = x(p) \ . \qedhere
\end{equation*}
\end{proof}
Given Claim \ref{cla:flowLPfeasible}, one might be tempted to run the intuitive randomized rounding as in previous works. It is easy to show that this rounding will select a tree in which every internal node will have between $(1-2/\log n)\cdot \log^{10} n$ and $(1+2/\log n)\cdot \log^{10} n$ children with high probability. Note that this implies that this tree will have at least $(1-2/\log n)^{h-1}\cdot (\log^{10} n)^{h-1}=\Omega(\log^{10(h-1)}n)$ leaves in layer $L_{h-1}$. These leaves will each demand $k= \Omega(\log^{10}(n))$ sinks. Hence the total number of sinks needed will be $\Omega(\log^{10h} n)=\Omega (n)$. We prove the following last claim. 

\begin{claim}
\label{cla:integraltree}
With high probability, any subtree $T$ rooted at $s$ with leaves at layer $h-1$ and in which every internal node has between  
\begin{equation*}
    \left(1-\frac{2}{\log n} \right) \log^{10}(n),
\end{equation*} 
and 
\begin{equation*}
    \left(1+\frac{2}{\log n} \right) \log^{10}(n)
\end{equation*}
many children is connected to at most 
\begin{equation*}
    O\left( \frac{(\log^{10} n)^h}{h} \right) = O \left( \frac{n}{h} \right)
\end{equation*} many sinks in layer $L_h$.
\end{claim}

Before giving a proof sketch of this, we explain why this shows that some kind of pruning is needed. As explained above, the intuitive randomized rounding will select a tree in which every node has roughly $\log^{10}(n)$ children with high probability. This tree needs to be connected to $\Omega (n)$ sinks in layer $L_h$ in order to give $\Omega(k)$ children to every leaf in $L_{h-1}$. By Claim \ref{cla:integraltree}, only a $O(1/h)$ fraction of these sinks are available to these leaves, yielding an approximation factor of $\Omega(h)=\Omega(\log n/\log \log n)$ on the vertices in $L_{h-1}$. Hence without any post-processing that can modify the tree, it seems impossible to break through this $\tilde \Omega (\log n)$ approximation factor.
Moreover, this instance shows that the post-processing cannot be only local: a simple post-processing 
one could think of is to remove vertices with high congestion
and hope that this does not remove too many children of any other vertex. However, it does not suffice to remove only children of $L_{h-1}$ in this example. Instead one needs to remove vertices from many layers.
A careful reader may notice that if instead of sampling $k$ children for every node (as it is the case in previous works) we allow ourselves to sample only $k/2$ children (which is equivalent to sample $k$ children and then perform our top-to-bottom pruning); then the necessary number of sinks in $L_h$ drops by a factor of $2^{-h}$. This factor seems to be enough to overcome the issue highlighted by our construction.

We finish this discussion by giving a proof sketch of Claim \ref{cla:integraltree}. 
\begin{proof}[Proof sketch for Claim \ref{cla:integraltree}]
Let us call $d$ the number of children of each node in a fixed tree $T$ and assume for simplicity that $d= \log^{10}(n)$ for all internal nodes in the tree. For any vertex $u\in L_j\cap T$ with $j\leq h-1$, we will denote by $T_u$ the subtree of $T$ rooted at $u$. We say that a vertex $v\in L_h$ \textit{survives} in $T_u$ if at least one leaf of $T_u$ is connected to $v$. We define
\begin{equation*}
      p_j:=\mathbb P(\textrm{$v$ survives in $T_u$} \mid \textrm{$v$ selected $u$}),
\end{equation*}
for any $u\in L_j\cap T$. Note that by symmetry this probability does not depend on which vertex $u\in T\cap L_j$ and which sink $v$ we select, hence the simplified notation.
We will now upper bound $p_j$ depending on the depth $j$. We note that for a vertex $v$ to survive in $T_u$, it must be that there is at least one child $u'$ of $u$ such that $v$ selects $u'$ and $v$ survives in $T_{u'}$. The probability that $v$ survives in $T_u$ conditioned by the fact that $v$ selected $u$ can be written as 
\begin{equation*}
    \mathbb P[\textrm{$v$ survives in $T_{u}$} \mid \textrm{$v$ selected $u$}]=1-\prod_{u' \textrm{ child of } u \textrm{ in the tree } T_u}\left(1-\frac{P[\textrm{$v$ survives in $T_{u'}$} \mid \textrm{$v$ selected $u'$}]}{\log^{10} n} \right)
\end{equation*} which by symmetry is equivalent to
\begin{equation*}
    p_{j}=1-\left(1-\frac{p_{j+1}}{\log^{10} n} \right)^{d}.
\end{equation*} 
When $n$ goes to infinity, this is roughly equal to 
\begin{equation*}
    1-e^{-\frac{d p_{j+1}}{\log^{10} n}}\approx 1-e^{-p_{j+1}}.
\end{equation*} 
From there, it is easy to prove by induction that 
\begin{equation*}
    p_j\leq \frac{O(1)}{h - j}.
\end{equation*}
Hence the expected number of vertices of the last layer that survive in any fixed tree $T$ rooted at the source is at most
\begin{equation*}
    p_0 \cdot n = O\left(\frac{n}{h}\right).
\end{equation*}
By independence and a standard Chernoff bound, one can show that with probability at most 
\begin{equation*}
    \exp\left( -\Omega\left( \frac{n}{h}\right)\right),
\end{equation*} more than $\frac{10n}{h}$ sinks survive in a fixed tree $T$. To upper bound the number of possible trees, note that it suffices so select the set of $(\log^{10}(n))^{h-1}$ leaves among the $(\log^{20}(n))^{h-1}$ possible leaves. Hence there are at most 
\begin{equation*}
    (\log^{20}(n))^{(h-1) \cdot (\log^{10} n)^{h-1}} = \exp\left(\Theta \left(\frac{n \cdot h \cdot \log\log (n)}{\log^{10} n}\right) \right) = \exp\left(o\left(\frac{n}{h} \right) \right)
\end{equation*} possible trees with internal degree $\log^{10}(n)$. By a standard union bound, this happens for no such tree.
\end{proof}


\section{Preprocessing the path LP}
\label{sec:Sparsify}
We recall the goal is to obtain, from a feasible fractional solution $x$ to the path LP 
a multiset of paths $P'$ containing $(s)$ for each $s\in S$
and satisfying
\begin{align}
    \sum_{q\in C_{P'}(p)} y(q) &= \frac{k}{4} \cdot y(p) &\forall p\in P'\setminus T'\label{eqSparseLP:APPdemand}\\
     \sum_{ q\in I_{P'}(v)\cap D_{P'}( p)} y( q) &\leq 2y( p) &\forall p\in P',v \in V\label{eqSparseLP:APPcapacity}\\
     y(p) &=  \frac{1}{(4\log^2 n)^i}  & \forall i\leq h, \forall  p\in L_i\label{eqSparseLP:APPgranularity}
\end{align}
The value of $y$ is fixed by the distance to the sources; hence, the only challenge in this part is in fact to select the multiset $P'$.

Towards this, consider a randomized rounding algorithm, that proceeds layer by layer starting from the sources. We will produce partial solutions $P^{(i)}$ that correspond to our final solution for all paths in $L_{\le i}$. Recall that this is a multiset and there might be several times the same copy of a path $p\in P$, but every copy will have a unique parent.
Initially, we add a single copy of the trivial path $(s)$ for each $s\in S$ to $P^{(0)}$.
Assume now that we defined the solutions up to $P^{(i)}$, that is, we sampled up to layer $i$ for some $i\ge 0$. For each path $p$ that is open and belongs to $P^{(i)}\cap L_i$, we sample $(k/4) \cdot (4 \log^2 n)$ many children paths where the $j$-th child equals path $q \in C_{P}(p) \cap L_{i+1}$ with probability
\begin{equation*}
    \frac{x(q)}{\sum_{q' \in C_{P}(p)} x(q')}\ .
\end{equation*}
For each path $q$ that is sampled, we set the parent of $q$ to be $p$ (which is a unique copy). Let $Q_q$ be the multiset of all copies of a path $q\in P$ that were sampled in this way. Then we set
\begin{equation*}
    P^{(i+1)}\leftarrow P^{(i)}\cup (\cup_{q\in L_{i+1}} Q_q).
\end{equation*}
We repeat this process until reaching the last layer $h$ and we set $P'=P^{(h)}$. We argue below that we obtain the desired properties
with high probability. To this end, we notice that each open path $p\in P^{(i)}\cap L_i$ samples exactly $(k/4) \cdot (4 \log^2 n)$ children paths $q$ and hence we obtain
\begin{equation*}
    \sum_{q\in C_{P'}(p)} y(q) = \frac{(k/4) \cdot (4 \log^2 n)}{(4 \log^2 n)^{i+1}} = \frac{k}{4} y(p) \ .
\end{equation*}
Therefore constraint~\eqref{eqSparseLP:APPdemand} is satisfied with probability~$1$. It remains to verify that constraint \eqref{eqSparseLP:APPcapacity} holds with high probability. For this, we consider the following quantity, that keeps track of the (possibly fractional) congestion induced by the descendants of a path $p\in P^{(i)}$ on a vertex $v$. If $v\in L_{\le i}$, then this induced congestion cannot change anymore so we will only keep track of this congestion for $v\in L_{\ge i+1}$. For ease of notation, let us define for all $q\in P$,
\begin{equation*}
    \textrm{cong}(q,v) := \sum_{q'\in D_P(q)\cap I_P(v)} x(q')\ .
\end{equation*}
Then, we define for any $p\in P^{(i)}$
\begin{equation*}
    \textrm{cong}(p,v \mid P^{(i)}) := \sum_{q\in D_{P^{(i)}}(p)\cap L_i} \frac{y(q)}{x(q)} \textrm{cong}(q,v)\ ,
\end{equation*}
which intuitively is the fractional congestion induced by descendants of $p$ on vertex $v$ but where we condition by what happened until layer $i$. Note that if $v\in L_j$ the quantity $\textrm{cong}(p,v \mid P^{(j)})$ is the final quantity we need to bound in order to ensure constraint \eqref{eqSparseLP:APPcapacity}. Indeed, in this case we obtain that 
\begin{equation*}
    \textrm{cong}(p,v \mid P^{(j)})=\sum_{q\in D_{P^{(j)}}(p) \cap I_{P^{(j)}}(v)} \frac{y(q)}{x(q)} \cdot x(q) = \sum_{q\in D_{P^{(j)}}(p) \cap I_{P^{(j)}}(v)} y(q) \ .
\end{equation*}
\begin{claim}
\label{cla:appendix_init_preprocessing}
For any $p\in P^{(i)}\cap L_i$, we have that 
\begin{equation*}
     \mathrm{cong}(p,v \mid P^{(i)}) \le y(p),
\end{equation*} with probability 1.
\end{claim}
\begin{proof}
If $p\in P^{(i)}\cap L_i$, then we have that 
\begin{equation*}
     \textrm{cong}(p,v \mid P^{(i)}) = y(p) \frac{\textrm{cong}(p,v)}{x(p)} \le y(p)\ ,
\end{equation*} using constraint \eqref{eqLP:flowLPcapacity} of the original path LP.
\end{proof}
\begin{claim}
\label{cla:appendix_deviation_preprocessing}
For any open $p\in P^{(i)}$ and any $i\ge 0$, we have that 
\begin{equation*}
     \mathbb P\left[\mathrm{cong}(p,v \mid P^{(i+1)}) \ge \frac{\mathrm{cong}(p,v \mid P^{(i)})}{2}+y(p)\right]\le \frac{1}{n^{2 \log (n)}} \ ,
\end{equation*} where the probability is taken over the randomness to round layer $L_{i+1}$.
\end{claim}
\begin{proof}
Denote by $N_q$ the number of copies of $q\in L_{i+1}\cap D_{P}(p)$ that are sampled as descendants of $p$. We also denote $P(q)$ the parent of $q$ in the set $P$. Accordingly, $N_{P(q)}$ is the number of copies of the path $P(q)$ that are chosen in $P^{(i)}$ as descendants of $p$. First, we compute
\begin{align*}
    \mathbb E\left[\textrm{cong}(p,v \mid P^{(i+1)})\right]&= \sum_{q\in  D_{P}(p)\cap L_{i+1}} \mathbb E[N_q] \cdot  \frac{1}{(4\log^2 n)^{i+1}\cdot x(q)} \cdot \textrm{cong}(q,v)\\
    &= \sum_{q\in D_{P}(p)\cap L_{i+1}} N_{P(q)}\cdot \frac{(k/4) \cdot (4 \log^2 n)}{(4\log^2 n)^{i+1}} \cdot \frac{x(q)}{\sum_{q' \in C_{P}(P(q))} x(q')} \cdot  \frac{\textrm{cong}(q,v)}{x(q)}\\
    &= \sum_{q\in D_{P}(p)\cap L_{i+1}} N_{P(q)}\cdot(k/4) \cdot \frac{1}{(4\log^2 n)^{i}} \cdot \frac{1}{\sum_{q' \in C_{P}(P(q))} x(q')} \cdot \textrm{cong}(q,v)\\
   &= \sum_{q\in D_{P}(p)\cap L_{i+1}} N_{P(q)}\cdot(k/4) \cdot y(P(q)) \cdot \frac{1}{k\cdot x(P(q))} \cdot \textrm{cong}(q,v)\\
    &=\frac{1}{4} \sum_{q\in D_{P}(p)\cap L_{i+1}} N_{P(q)}\cdot \frac{y(P(q))}{x(P(q))} \cdot \textrm{cong}(q,v)\\
    &= \frac{1}{4} \sum_{q'\in D_{P}(p)\cap L_{i}}N_{q'}\cdot  \frac{y(q')}{x(q')} \cdot \sum_{q\in C_P(q')}\textrm{cong}(q,v)\\
    &= \frac{1}{4} \sum_{q'\in D_{P}(p)\cap L_{i}} N_{q'}\cdot\frac{y(q')}{x(q')} \cdot \textrm{cong}(q',v) = \frac{\textrm{cong}(p,v \mid P^{(i)})}{4}.
\end{align*} where we used constraint \eqref{eqLP:flowLPdemand} of the original path LP to obtain the fourth line, and standard algebraic manipulations for the rest. Second, we note that $\textrm{cong}(p,v \mid P^{(i+1)})$ can be written as a sum of independent random variables of value 
\begin{equation*}
    \frac{y(q)}{x(q)} \cdot \textrm{cong}(q,v)
\end{equation*} for some $q\in L_{i+1}$. By constraint \eqref{eqLP:flowLPcapacity} of the path LP, the absolute value of these variables is never more than $1 / (100\log^2 n)^{i+1}$. By a standard Chernoff bound, we can conlude that 

\begin{align*}
    \mathbb P\left[\mathrm{cong}(p,v \mid P^{(i+1)}) \ge \frac{\mathrm{cong}(p,v \mid P^{(i)})}{2}+y(p)\right]
    &\le\exp \left(-\frac{y(p)(4\log^2 n)^{i+1}}{2} \right)  \\
    &\le \exp \left(-2\log^2 n \right) \le n^{-2\log(n)} \ .\qedhere
\end{align*}
\end{proof}
There are at most $n^h\le n^{\log(n)+1}$ constraints to maintain to ensure constraint~\eqref{eqSparseLP:APPcapacity} over at most $h$ rounds of rounding. Hence by Claim~$\ref{cla:appendix_deviation_preprocessing}$, we have $\mathrm{cong}(p,v \mid P^{(i+1)}) \le \mathrm{cong}(p,v \mid P^{(i)})/2+y(p)$ for all paths $p$, vertices $v$, and rounds $i$ with probability at least $1-n^{-\log(n)+1}$. Using Claim \ref{cla:appendix_init_preprocessing} with this fact we obtain that for any path $p$ and vertex $v$
\begin{equation*}
    \sum_{ q\in I_{P'}(v)\cap D_{P'}( p)} y( q) \le y(p) \sum_{j=0}^{\infty} \frac{1}{2^j} \le 2 y(p).
\end{equation*}
This proves the desired result.

\section{APX-hardness of max-min degree arborescence}
\label{sec:hardness}
We show here that the max-min degree arborescence problem is APX-hard, even in a layered graph with 2 layers and a single source. Our proof is based on a reduction from the max-$k$-cover problem, in which one is given a universe $\mathcal U$ of $m$ elements, a family of subsets $S_1,S_2,\ldots ,S_n$ of subsets of $\mathcal U$ and one has to cover a maximum number of elements of $\mathcal U$ using at most $k$ sets. In a seminal result by Feige~\cite{feige1998threshold}, it is showed that it is NP-hard to approximate max-$k$-cover within factor better than $\frac{e}{e-1}$. More specifically, it is showed that on instances where all sets have the same cardinality $m/k$, it is NP-hard to decide whether there exist $k$ disjoint sets that cover the whole universe $\mathcal U$ (\textbf{YES} instance) or if any $k$ sets cover at most $(1-1/e)\cdot m$ elements (\textbf{NO} instance).

\begin{theorem}
For any $\epsilon> 0$ there is no $(\sqrt{\frac{e}{e-1}}-\epsilon)$-approximation algorithm to the single source max-min degree arborescence problem on 2-layered graphs, unless P=NP.
\end{theorem}
\begin{proof}
From a max-$k$-cover instance $(\mathcal U,S_1,S_2,\ldots, S_n)$ we make a layered instance of max-min arborescence as follows. In layer $L_0$ there is a single source $s$. In layer $L_1$, there are $m/k^2$ vertices $(v_{i,j})_{1\le j\le m/k^2}$ for each set $S_i$. In layer $L_2$, there are $m/k^2$ sinks $(s_{a,j})_{1\le j\le m/k^2}$ for each element $a\in \mathcal U$.

Then we connect the source $s$ to all vertices in layer $L_1$. Each vertex $v_{i,j}$ in $L_1$ is connected to all sinks $s_{a,j}$ such that $a\in S_i$.
Let us denote by $\textrm{OPT}$ the optimum max-min degree of an arborescence on this instance.

\begin{claim}
If the max-$k$-cover instance is a \textbf{YES} instance, then $\textrm{OPT}\ge m/k$. Otherwise, $\textrm{OPT}\le \sqrt{1-1/e}\cdot m/k$.
\end{claim}
\begin{proof}
If the max-$k$-cover instance is a \textbf{YES} instance then for the source we select as neighbors in $L_1$ all the vertices $v_{i,j}$ such that $S_i$ is selected in the $k$-cover solution. The source gets in this way $k\cdot m/k^2=m/k$ outgoing neighbors. Each vertex $v_{i,j}$ that was selected then selects all the sinks $s_{a,j}$ such that $a\in S_i$. Since we are in a \textbf{YES} instance, all the selected sets are disjoint and each one receives exactly $m/k$ sinks. This is a valid arborescence of value $m/k$. 

In the \textbf{NO} case, consider any valid arborescence of value \textrm{OPT}. If there is a $j$ such that the source selects more than $(1+\beta)\cdot k$ vertices (for any $\beta\ge 0$) in the set $V_j=(v_{i,j})_{1\le i \le m}$ then this arborescence cannot have a maxmin degree more than 
\begin{equation*}
   \max_{\beta\ge 0} \min \left\lbrace \frac{1-1/e+\beta}{1+\beta},\frac{1}{1+\beta} \right\rbrace \cdot \frac{m}{k} \le \frac{1}{1+1/e}\cdot \frac{m}{k} \le \left(\sqrt{1-1/e}\right)\cdot \frac{m}{k}
\end{equation*} To see this, note that in a \textbf{NO} instance, any union of $k$ sets contains at most $(1-1/e)\cdot m$ elements, hence any union of $(1+\beta)\cdot k$ sets contains at most $(1-1/e)\cdot m+ \beta m$ elements (recall that each set has cardinality $m/k$), to be shared between $(1+\beta)\cdot k$ vertices. This gives the first ratio. The second comes from the fact that any union of some sets contains at most $m$ elements. Optimizing over all $\beta\ge 0$
 gives the upper-bound. Hence assume there is not such $j$.

If the source $s$ has outdegree more than $\sqrt{1-1/e}\cdot m/k$ then it must be that there is a $j\in [1,m/k^2]$ such that $s$ has more than
\begin{equation*}
    \frac{\sqrt{1-1/e}\cdot m/k}{m/k^2}=\sqrt{1-1/e}\cdot k
\end{equation*}
out-neighbors in the set $V_j=(v_{i,j})_{1\le i\le m}$. By assumption, the source has also at most $k$ out-neighbors in this set hence the number of sinks connected to vertices in $V_j$ is at most $(1-1/e)\cdot m$ (recall we are in a \textbf{NO} case). Hence there must be a vertex $v_{i,j}$ that gets less than 
\begin{equation*}
    \frac{(1-1/e)\cdot m}{\sqrt{1-1/e}\cdot k} = \left(\sqrt{1-1/e}\right)\cdot \frac{m}{k}
\end{equation*} sinks as neighbors. In all cases, there must be a vertex in the arborescence that gets out-degree at most $(\sqrt{1-1/e})\cdot \frac{m}{k}$, which ends the proof of the claim.
\end{proof}
By the previous claim, a $(\sqrt{\frac{e}{e-1}}-\epsilon)$-approximation would be able to distinguish between the \textbf{YES} and \textbf{NO} instances of the max-$k$-cover problem, which is NP-hard.
\end{proof}

\end{document}